%% file: main.tex
\documentclass[11pt]{article}

\usepackage{microtype}
\usepackage[margin=1.in]{geometry}
\usepackage{amsmath}
\usepackage{amssymb}
\usepackage{amsthm}
\usepackage{ifpdf}
\usepackage{wrapfig}
\usepackage{fullpage}
\usepackage{cite}
\usepackage{float}
\usepackage{graphicx}
\usepackage{subcaption}
\usepackage{comment}

\newcommand{\BO}[1]{\mathcal{O}(#1)}
\newcommand{\N}{\mathbb{N}}

\usepackage[pdftex]{epsfig}
\usepackage[pdftex]{hyperref}

\newtheorem{theorem}{Theorem}[section]
\newtheorem{lemma}[theorem]{Lemma}

\newtheorem{definition}[theorem]{Definition}

\begin{document}

\title{Optimal Staged Self-Assembly of General Shapes\footnote{Research supported in part by National Science Foundation Grants CCF-1117672, CCF-1555626, and CCF-1422152.  An abstract version of this work appeared as~\cite{CMS2016OSS}.}}
\author{
Cameron Chalk
\thanks{Department of Computer Science, University of Texas - Rio Grande Valley, 
\protect\url{cameron.chalk01@utrgv.edu}, \protect\url{eric.m.martinez02@utrgv.edu}, \protect\url{robert.schweller@utrgv.edu},
\protect\url{luis.a.vega01@utrgv.edu}, \protect\url{andrew.winslow@utrgv.edu}, \protect\url{timothy.wylie@utrgv.edu}}
\and
Eric Martinez\footnotemark[2]
\and
Robert Schweller\footnotemark[2]
\and
Luis Vega\footnotemark[2]
\and
Andrew Winslow\footnotemark[2]~\footnote{A substantial portion of the work of this author was performed while at Universit\'{e} Libre de Bruxelles.}
\and
Tim Wylie\footnotemark[2]
}
\date{}

\maketitle

\begin{abstract}
We analyze the number of tile types $t$, bins $b$, and stages necessary to assemble $n \times n$ squares and scaled shapes in the staged tile assembly model.
For $n \times n$ squares, we prove $\BO{\frac{\log{n} - tb - t\log t}{b^2} + \frac{\log \log b}{\log t}}$ stages suffice and $\Omega(\frac{\log{n} - tb - t\log t}{b^2})$ are necessary for almost all $n$.
For shapes $S$ with Kolmogorov complexity $K(S)$, we prove $\BO{\frac{K(S) - tb - t\log t}{b^2} + \frac{\log \log b}{\log t}}$ stages suffice and $\Omega(\frac{K(S) - tb - t\log t}{b^2})$ are necessary to assemble a scaled version of $S$, for almost all $S$.
We obtain similarly tight bounds when the more powerful \emph{flexible} glues are permitted.
\end{abstract}

\input{introduction}

\input{model}
\input{key_lemmas}

\input{bit_string_pad_construction}

\input{nxn_square}

\input{scaled_shapes}

\input{flexible_glues}
\input{conclusion}

\bibliographystyle{plain}
\bibliography{main}

\end{document}

%% file: introduction.tex
\section{Introduction}\label{sec:introduction}

\emph{Self-assembly} is the process by which many small, simple particles attach via local interactions to form large, complex structures.
In the early 1980s, Ned Seeman~\cite{Seem82} demonstrated that short strands of DNA, attaching via matching base pairs, could be ``programmed'' to carry out such self-assembly at the nanoscale.
In the 1990s, the Ph.D. thesis of Erik Winfree~\cite{Winf98} introduced an approach for DNA self-assembly wherein four-sided DNA ``tiles'' attach edgewise via matching sequences of base pairs.

Such \emph{DNA tile self-assembly} is both experimentally~\cite{ConstantineThesis} and theoretically~\cite{Winf98,RotWin00} capable of complex behaviors.
Study of theoretical models of such tile assembly, beginning with the \emph{abstract Tile Assembly Model (aTAM)} of Winfree~\cite{Winf98}, has grown into a field with dozens of model variations~\cite{Doty-2012a,nacoTileSurvey} compared via a web of reductions~\cite{Woods2015}, reminiscent of the development of structural complexity for traditional models of computation.

\textbf{Staged tile assembly.}
The \emph{staged (tile assembly) model} is a generalization of the \emph{two-handed (2HAM)}~\cite{RNaseSODA2010,2HABTO,DPR2013ICALP,DDFIRSS07} or \emph{hierarchical}~\cite{CheDot12,Doty16} \emph{tile assembly model}.
In the aTAM, single tiles attach to a growing multi-tile seed assembly.
The 2HAM additionally permits multi-tile assemblies attach to each other, yielding a growth process strictly more general than that of the aTAM~\cite{2HABTO}.

The staged assembly model extends the 2HAM by including the ability to simultaneously carry out separate assembly processes in multiple \emph{bins}.
Such separation is motivated by similar experimental practices wherein many reactions are carried out simultaneously in distinct test tubes, often automated via liquid-handling robots~\cite{Ke1177}.

In the staged model, a bin begins with a set of \emph{input assemblies} previously assembled in other bins.
These assemblies are repeatedly attached pairwise to yield a growing set of \emph{producible assemblies} until all possibilities are exhausted.
The producible assemblies not attachable to any other producible assemblies during this process are the \emph{output assemblies} of the bin, and may be used as input assemblies for other bins.

An instance of the staged model, called a \emph{staged system}, consists of many bins stratified into stages, and a \emph{mix graph} that specifies which bins in each stage supply input assemblies for bins in the following stage.
Input assemblies for bins in the first stage are sets of individual tiles, and the output assemblies of bins in the final stage are considered the \emph{output} of the system.

\textbf{Complexity in staged assembly.}
A common goal in the design of self-assembling systems is the assembly of a desired shape.
Here we consider the design of ``efficient'' systems with minimum complexity that assemble a given shape.
Three metrics exist for staged systems:
\begin{itemize}
\item \emph{Tile type complexity:} the number of distinct tile types used in the system.
\item \emph{Bin complexity:} the maximum number of bins used in a stage.
\item \emph{Stage complexity:} the number of stages.
\end{itemize}
Efficient assembly of restricted classes of shapes~\cite{DDFIRSS07,DFS2015NGA} and $1 \times n$ ``patterned lines''~\cite{demaine2013ODS,W2015SSA} have been considered, and efficient assembly in other variants and further extensions and variants of the staged self-assembly model have also been studied~\cite{RNaseSODA2010,BMS2012TUS,RNAPods,Labean05stepwisedna,MSS2012SWT,IDSUSA}.

Two classes of shapes have become the defacto standards~\cite{Winslow-2016b} for measuring assembly efficiency: squares and and general shapes (with scaling permitted).
Efficient assembly of these two classes was first considered in the aTAM, where matching upper and lower bounds on the tile complexity were obtained~\cite{RotWin00,SolWin07,AdChGoHu01}.

\textbf{Our results.}
Here we give nearly matching upper and lower (trivariate) bounds for assembling these shapes in the staged model; our results are summarized here and in Table~\ref{tab:summary}.
Due to the multivariate nature of the results, they are described using two complexity measures (tile type and bin) as ``independent'' variables and the third measure (stage) as the ``dependent'' variable.
That is, the results are upper and lower bounds on the minimum stage complexity of systems with a fixed number of tile types $t$ and bins $b$, where $t$ and $b$ are restricted to be larger than some fixed constant.\footnote{Such a restriction is necessary, as systems with a single tile type are incapable of assembling finite non-trivial shapes.}

For $n \times n$ squares, we prove the stage complexity is $\BO{\frac{\log{n} - tb - t\log t}{b^2} + \frac{\log \log b}{\log t}}$ and, for almost all $n$, $\Omega(\frac{\log{n} - tb - t\log t}{b^2})$.\footnote{The fraction of values for which the statement holds reaches~1 in the limit as $n \rightarrow \infty$.}
For shapes $S$ with Kolmogorov complexity $K(S)$, we prove the stage complexity is $\BO{\frac{K(S) - tb - t\log t}{b^2} + \frac{\log \log b}{\log t}}$ and 
$\Omega(\frac{K(S) - tb - t\log t}{b^2})$.


We obtain similar results when \emph{flexible glues}~\cite{AGKS05g}, glues that can form bonds with non-matching glues, are permitted.
In this case, the stage complexity for $n \times n$ squares is reduced to $\BO{\frac{\log n - t^2 - tb}{b^2} + \frac{\log \log b}{\log t}}$ and, for almost all $n$, $\Omega(\frac{\log n - t^2 - tb}{b^2})$, and the stage complexity for general shapes is reduced to $\BO{\frac{K(S) - t^2 - tb}{b^2} + \frac{\log \log b}{\log t}}$ and $\Omega(\frac{K(S) - t^2 - tb}{b^2})$.

Our upper bounds both use a new technique to efficiently assemble \emph{bit string pads}: constant-width assemblies with an exposed sequence of glues encoding a desired bit-string.
This technique converts all three forms of system complexity (tile, bin, and stage) into bits of the string with only a constant-factor loss of efficiency.
In other words, the number of bits in the bit string pad grows proportionally to the number of bits needed to describe the system.

\textbf{Comparison with prior work.}
Because our results are optimal across all choices of tile type and bin complexity, these results generalize and, in some cases, improve on prior results.
For instance, Theorem~\ref{thm:nxn} implies assembly of $n \times n$ squares using $\BO{1}$ bins, $\BO{\frac{\log n}{\log\log n}}$ tile types, and $\BO{1}$ stages, matching a result of~\cite{AdChGoHu01} up to constant factors.
For flexible glues, this is improved to $\BO{\sqrt{\log n}}$ tile types, a result of~\cite{AGKS05g}.

The same theorem also yields assembly of $n \times n$ squares by systems with $\BO{1}$ bins, $\BO{1}$ tile types, and $\BO{\log{n}}$ stages (a result of~\cite{DDFIRSS07}) and by systems with $\BO{\sqrt{\log{n}}}$ bins, $\BO{1}$ tile types, and $\BO{\log\log\log{n}}$ stages, substantially improving over the $\mathcal{O}(\log\log n)$ stages used in~\cite{DDFIRSS07}.
For constructing scaled shapes, Theorem~\ref{thm:shapes} implies systems using $\BO{1}$ bins, $\BO{\frac{K(S)}{\log K(S)}}$ tile types, and $\BO{1}$ stages, a result of~\cite{SolWin07}.

\begin{table}[ht]
\centering
\renewcommand{\arraystretch}{1.3}
    \begin{tabular}{| c | c | c | c | c |}
    \multicolumn{5}{c}{\textbf{Standard Glue Stage Complexity Results}} \\ \hline
      \textbf{Shape}                    & \textbf{Upper Bound}                      & \textbf{Theorem} & \textbf{Lower Bound} & \textbf{Theorem} \\
      \hline
      $n \times n$   & $\BO{\frac{\log n - t\log t - tb}{b^2} + \frac{\log \log b}{\log t}}$    & \ref{thm:nxn} & $\Omega(\frac{\log n - t\log t - tb}{b^2})$ & \ref{thm:nxnLower} \\
      \hline
      Scaled shapes                     & $\BO{\frac{K(S) - t\log t - tb}{b^2} + \frac{\log \log b}{\log t}}$    & \ref{thm:shapes} & $\Omega(\frac{K(S) - t\log t - tb}{b^2})$ & \ref{thm:shapesLower} \\
      \hline
      \multicolumn{5}{c}{\textbf{Flexible Glue Stage Complexity Results}} \\
      \hline

      $n \times n$   & $\BO{\frac{\log n - t^2 - tb}{b^2} + \frac{\log \log b}{\log t}}$    & \ref{thm:nxnFlex} & $\Omega(\frac{\log n - t^2 - tb}{b^2})$ & \ref{thm:nxnLowerFlex} \\
      \hline
      Scaled shapes                     & $\BO{\frac{K(S) - t^2 - tb}{b^2} + \frac{\log \log b}{\log t}}$    & \ref{thm:shapesFlex} & $\Omega(\frac{K(S) - t^2 - tb}{b^2})$ & \ref{thm:shapesLowerFlex} \\
      \hline
    \end{tabular}
    \caption{The main results obtained in this work: upper and lower bounds on the number of stages of a staged self-assembly system with $b$ bins and $t$ tile types uniquely assembling $n \times n$ squares and scaled shapes. $K(S)$ denotes the Kolmogorov complexity of a shape.}
    \label{tab:summary}
\end{table}

%% file: model.tex
\section{The Staged Assembly Model}\label{sec:model}

\textbf{Tiles.}
A \emph{tile} is a non-rotatable unit square with each edge labeled with a \emph{glue} from a set $\Sigma$.
Each pair of glues $g_1, g_2 \in \Sigma$ has a non-negative integer \emph{strength}, denoted ${\rm str}(g_1, g_2)$.
Every set $\Sigma$ contains a special \emph{null glue} whose strength with every other glue is 0.
If the glue strengths do not obey ${\rm str}(g_1,g_2) = 0$ for all $g_1 \neq g_2$, then the glues are \emph{flexible}.
Unless otherwise stated, we assume that glues are not flexible.

\textbf{Configurations, assemblies, and shapes.}
A \emph{configuration} is a partial function $A : \mathbb{Z}^2 \rightarrow T$ for some set of tiles $T$, i.e., an arrangement of tiles on a square grid.
For a configuration $A$ and vector $\vec{u} = \langle u_x, u_y \rangle \in \mathbb{Z}^2$, $A + \vec{u}$ denotes the configuration $f \circ A$, where $f(x, y) = (x + u_x, y + u_y)$.
For two configurations $A$ and $B$, $B$ is a \emph{translation} of $A$, written $B \simeq A$, provided that $B = A + \vec{u}$ for some vector $\vec{u}$.
For a configuration $A$, the \emph{assembly} of $A$ is the set $\tilde{A} = \{ B : B \simeq A \}$.
An assembly $\tilde{A}$ is a \emph{subassembly} of an assembly $\tilde{B}$, denoted $\tilde{A} \sqsubseteq \tilde{B}$, provided that there exists an $A\in \tilde{A}$ and $B\in \tilde{B}$ such that $A \subseteq B$.
The \emph{shape} of an assembly $\tilde{A}$ is $\{ {\rm dom}(A) : A \in \tilde{A}\}$ where dom() is the domain of a configuration.
A shape $S'$ is a \emph{scaled} version of shape $S$ provided that for some $k \in \N$ and $D \in S$, $\bigcup_{(x, y) \in D} \bigcup_{(i, j) \in \{0, 1, \dots, k-1\}^2} (kx + i, ky + j) \in S'$.

\textbf{Bond graphs and stability.}
For a configuration~$A$, define the \emph{bond graph}~$G_A$ to be the weighted grid graph in which each element of~${\rm dom}(A)$ is a vertex, and the weight of the edge between a pair of tiles is equal to the strength of the coincident glue pair.
A configuration is \emph{$\tau$-stable} for $\tau \in \N$ if every edge cut of $G_A$ has strength at least $\tau$, and is \emph{$\tau$-unstable} otherwise.
Similarly, an assembly is \emph{$\tau$-stable} provided the configurations it contains are $\tau$-stable.
Assemblies $\tilde{A}$ and $\tilde{B}$ are \emph{$\tau$-combinable} into an assembly $\tilde{C}$ provided there exist $A \in \tilde{A}$, $B \in \tilde{B}$, and $C \in \tilde{C}$ such that $A \bigcup B = C$, $\rm dom(A) \bigcap \rm dom(B) = \emptyset$, and $\tilde{C}$ is $\tau$-stable.

\textbf{Two-handed assembly and bins.}
We define the assembly process via bins.
A bin is an ordered tuple $(S,\tau)$ where $S$ is a set of \emph{initial} assemblies and $\tau \in \N$ is the \emph{temperature}.
In this work, $\tau$ is always equal to $2$ for upper bounds, and at most some constant for lower bounds.
For a bin $(S, \tau)$, the set of \emph{produced} assemblies $P'_{(S,\tau)}$ is defined recursively as follows:
\begin{enumerate}
\item $S \subseteq P'_{(S,\tau)}$.
\item If $A,B \in P'_{(S,\tau)}$ are $\tau$-combinable into $C$, then $C \in P'_{(S,\tau)}$.
\end{enumerate}
A produced assembly is \emph{terminal} provided it is not $\tau$-combinable with any other producible assembly, and the set of all terminal assemblies of a bin $(S,\tau)$ is denoted $P_{(S,\tau)}$.
That is, $P'_{(S,\tau)}$ represents the set of all possible assemblies that can assemble from the initial set $S$, whereas $P_{(S,\tau)}$ represents only the set of assemblies that cannot grow any further.

The assemblies in $P_{(S,\tau)}$ are \emph{uniquely produced} iff for each $x \in P'_{(S, \tau)}$ there exists a corresponding $y \in P_{(S,\tau)}$ such that $x \sqsubseteq y$.
Unique production implies that every producible assembly can be repeatedly combined with others to form an assembly in $P_{(S,\tau)}$.

\textbf{Staged assembly systems.}
An \emph{$r$-stage $b$-bin mix graph} $M$ is an acyclic $r$-partite digraph consisting of $rb$ vertices $m_{i,j}$ for $1 \leq i \leq r$ and $1\leq j \leq b$, and edges of the form $(m_{i,j}, m_{i+1,j'})$ for some $i, j, j'$.
A \emph{staged assembly system} is a 3-tuple $\langle M_{r,b}, \{T_1, T_2, \dots, T_b\}, \tau \rangle$ where $M_{r,b}$ is an $r$-stage $b$-bin mix graph, $T_i$ is a set of tile types, and $\tau \in \N$ is the temperature.
Given a staged assembly system, for each $1\leq i \leq r$, $1\leq j \leq b$, a corresponding bin $(R_{i,j}, \tau)$ is defined as follows:

\begin{enumerate}
\item $R_{1,j}= T_j$ (this is a bin in the first stage);
\item For $i\geq 2$,
$\displaystyle R_{i,j}= \Big(\bigcup_{k:\ (m_{i-1,k},m_{i,j})\in M_{r,b}} P_{(R_{(i-1,k)},\tau_{i-1,k})}\Big)$.
\end{enumerate}

Thus, bins in stage 1 are tile sets $T_j$, and each bin in any subsequent stage receives an initial set of assemblies consisting of the terminally produced assemblies from a subset of the bins in the previous stage as dictated by the edges of the mix graph.\footnote{The original staged model~\cite{DDFIRSS07} only considered $\BO{1}$ distinct tile types, and thus for simplicity allowed tiles to be added at any stage (since $\mathcal{O}(1)$ extra bins could hold the individual tile types to mix at any stage). Because systems here may have super-constant tile complexity, we restrict tiles to only be added at the initial stage.}
The \emph{output} of a staged system is the union of the set of terminal assemblies of the bins in the final stage.\footnote{This is a slight modification of the original staged model~\cite{DDFIRSS07} in that there is no requirement of a final stage with a single output bin. It may be easier in general to solve problems in this variant of the model, so it is considered for lower bound purposes. However, all results herein apply to both variants of the model.}
The output of a staged system is \emph{uniquely produced} provided each bin in the staged system uniquely produces its terminal assemblies.

%% file: key_lemmas.tex
\begin{figure}[t!]
    \centering
    \hspace*{-.90cm}
    \subcaptionbox{}{\includegraphics[width=0.47\textwidth]{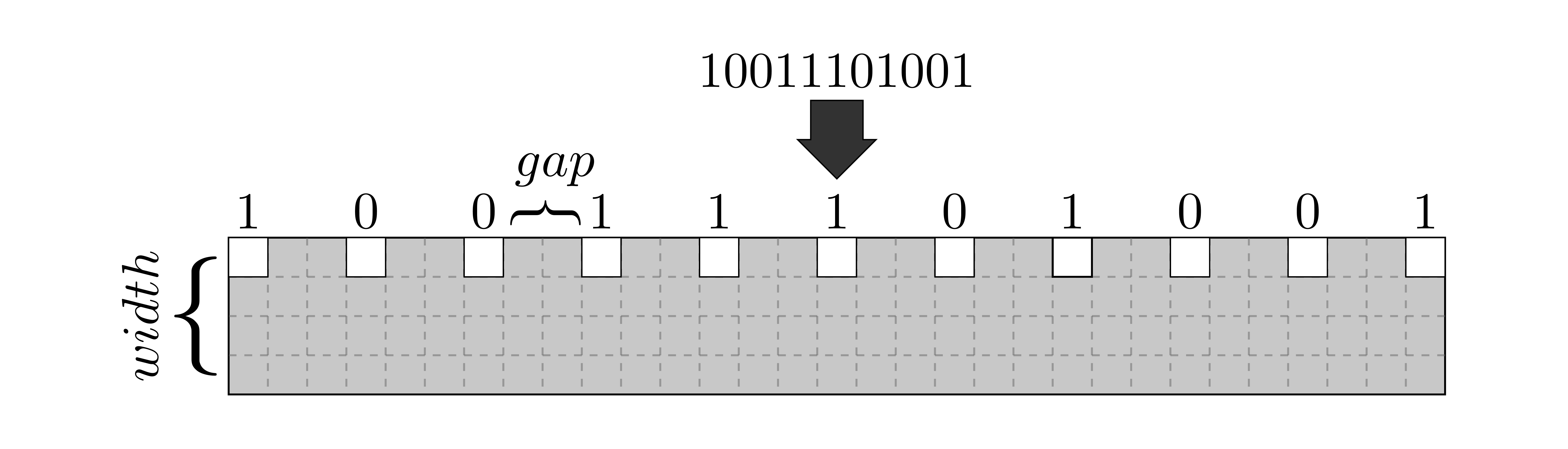}}
    \hspace*{.40cm}
    \subcaptionbox{}{\includegraphics[width=0.47\textwidth]{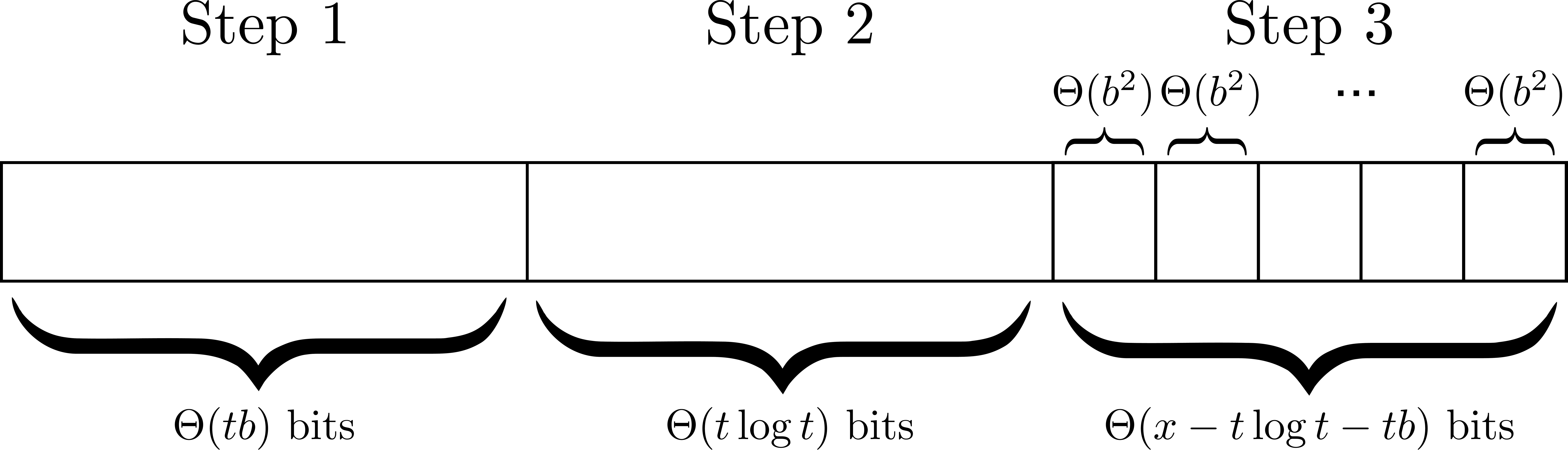}}
    \caption{For compactness, glues are denoted by their subscripts in figures for the remainder of the paper (e.g. $g_0$ denoted as $0$).
    (a) An example bit string $r=10011101001$ encoded as a width-$4$ gap-$2$ $11$-\emph{bit string pad} whose north-facing glues correspond to the bits in $r$.
    (b) The decomposition of an $x$-bit string pad's bits into those encoded by the three steps of a staged system with $t$ tile types and $b$ bins.}
    \label{fig:mt_overview}
\end{figure}

\section{Key Lemmas} \label{sec:key_lemmas}

Our results rely on two key lemmas.
The first proves an upper bound on the information content of a staged system, and thus a lower bound on the stage complexity for assembling shapes containing a specified quantity of information.
The second is a formal statement of the previously mentioned bit string pad construction.

\begin{lemma}
\label{lem:info-lowerbound}
A staged system of fixed temperature $\tau$ with $b$ bins, $s$ stages, and $t$ tile types can be specified using $\BO{t\log{t} + sb^2 + tb}$ bits.
Such a system with flexible glues can be specified using $\BO{t^2 + sb^2 + tb}$ bits.
\end{lemma}

\begin{proof}
A staged system can be specified in four parts: the tile types, the glue function, the mix graph, and the assignment of tile types to stage-1 bins.
We separately bound the number of bits required to specify each.

A set of $t$ tile types has up to~$4t$ glue types, so specifying each tile requires $\BO{\log{t}}$ bits, and the entire tile set takes $\BO{t\log{t}}$ bits.
If the system does not have flexible glues, then the glue function can be specified in $\BO{4t} = \BO{t}$ bits, using $\BO{\log{\tau}} = \BO{1}$ bits per glue type to specify the glue's strength.
If the system has flexible glues, then the glue function can be specified using $\BO{1}$ bits per pairwise glue interaction and $\BO{(4t)^2} = \BO{t^2}$ bits total.

The mix graph consists of~$bs$ nodes.
Each pair of nodes in adjacent stages optionally share a directed edge.
Thus specifying these edges takes $\BO{b^2(s-1)} = \BO{b^2s}$ bits.
The assignment of tile types to stage-1 bins requires one bit per each choice of tile type and bin, or $\BO{tb}$ bits total.

Thus a staged system without flexible glues can be specified in $\BO{t\log{t} + t + b^2s + tb}$ bits, and otherwise in $\BO{t\log{t} + t^2 + b^2s + tb}$ bits.
\end{proof}

It immediately follows from Lemma~\ref{lem:info-lowerbound} that for almost all bit strings, any staged system with $b$ bins and $t$ tiles that encodes the bit string must have $\Omega(\frac{x-tb - t\log t }{b^2})$ stages with standard glues and $\Omega(\frac{x - tb - t^{2}}{b^2})$ stages with flexible glues, where $x$ is the length of the bit string.

Our positive results rely mainly on efficient assembly of \emph{bit string pads}: assemblies that expose a sequence of north-facing glues that encode a bit string. An example is shown in Figure~\ref{fig:mt_overview}(a).
Squares and general scaled shapes are assembled by combining a universal set of ``computation'' tiles with efficiently assembled ``input'' bit string pads.

\begin{definition}[Bit string pad]
A \emph{width-$k$ gap-$f$ $r$-bit string pad} is a $k\times (f(r-1)+1)$ rectangular assembly with $r$ glues from a set of two glue types $\{g_0, g_1\}$ exposed on the north face of the rectangle at intervals of length $f$, starting from the westmost northern edge.
All remaining exposed glues on the north tile edges have some common label $g_F$.
The remaining exposed south, east, and west tile edges have glues $g_S$, $g_E$, and $g_W$.
A bit string pad \emph{represents} a given string of $r$ bits if the exposed $g_0$ and $g_1$ glues from west to east, on the north facing edges, are equal to the given bit string.
\end{definition}

For our upper bounds we efficiently construct bit string pads by decomposing the target pad into three subpads and assembling each in a separate step using a different source of system complexity (see Figure~\ref{fig:mt_overview}(b)):
\begin{itemize}
\item Step~1: $\Theta(tb)$ bits from assigning tile types to stage-1 bins (Section~\ref{sec:step_1}).
\vspace*{-.1cm}
\item Step~2: $\Theta(t\log{t})$ bits from the tile types themselves as in~\cite{AdChGoHu01,AGKS05g,FMS2015OPS,SolWin07} (Section~\ref{sec:step_2}).
\vspace*{-.1cm}
\item Step~3: $\Theta(x - t\log{t} - tb)$ bits from the mix graph using a variant of ``crazy mixing''~\cite{DDFIRSS07} (Section~\ref{sec:step_3}).
\end{itemize}
If flexible glues are permitted, Step~2 is modified as in~\cite{AGKS05g} to achieve $\Theta(t^2)$ bits from tile types.
These subpads are then combined into the complete pad to achieve the following result:

\begin{lemma}
\label{lem:xbit_string}
There exists a constant $c$ such that for any $b,t \in \mathbb{N}$ with $b,t > c$ and any bit string $S$ of length $x$,
there exists a $\tau = 2$ staged assembly system with $b$ bins, $t$ tiles, and $\mathcal{O}(\frac{x-tb-t\log t}{b^2} + \frac{\log \log b}{\log t})$ stages whose uniquely produced output is a width-$7$ gap-$\Theta(\log{b})$ $x$-bit string pad representing $S$.
\end{lemma}

\begin{proof}
Consider a bit string $S$ of length $x$.
The width-7 gap-$\Theta(\log{b})$ $x$-bit string pad representing $S$ to be constructed is referred to as the \emph{target bit string pad}, and the staged assembly system which assembles this pad is the \emph{target system}.
The techniques used in the target system require that the tile types and bins must be greater than a constant $c$.
The target system is designed by allotting a constant fraction of the system's total tile types and bins (i.e. $\frac{t}{c}$ and $\frac{b}{c}$ for some constant $c$) to each of four steps .
The first three steps each assemble a subpad of the target bit string pad, which are concatenated in a final step. The steps' subpads and stage complexities are as follows:
\begin{itemize}
\item Step~1: A width-7 gap-$\Theta(\log b)$ $\Theta(tb)$-bit string pad representing $\Theta(tb)$ bits of $S$. Uses $\BO{\frac{\log \log b}{\log t}}$ stages (see Lemma~\ref{lem:step_1}).
\item Step~2: A width-3 gap-$\Theta(\log b)$ $\Theta(t\log t)$-bit string pad representing $\Theta(t \log t)$ bits of $S$. Uses $\BO{\frac{\log \log b}{\log t}}$ stages (see Lemma~\ref{lem:step_2_2} and Section~\ref{sec:gapped_step_2}).
\item Step~3: A width-7 gap-$\Theta(\log b)$ $(x-\Theta(tb)-\Theta(t\log t))$-bit string pad representing the remaining $(x-\Theta(tb)-\Theta(t\log t))$ bits of $S$. Uses $\mathcal{O}(\frac{x-tb-t\log t}{b^2} + \frac{\log \log b}{\log t})$ stages (see Lemma~\ref{lem:step_3}).
\end{itemize}

The subpads constructed by these steps cannot simply be concatenated, as their adjacent bits at the point of concatenation would not have the same uniform gap as the other bits.
A length-$\Theta(\log b)$ \emph{connector} must be assembled to concatenate the subpads while inducing the appropriate gap.
Lemma~\ref{lem:wings}, although describing a system which assembles an assembly with more requirements than needed for this purpose, proves the existence of a system with $t$ tile types and $b$ bins which can assemble a length-$\Theta(\log b)$ assembly using $\mathcal{O}(\frac{\log\log b}{\log t})$ stages.
Such a system can trivially be modified to expose glues on east and west ends for use in concatenating the subpads.
In one stage, concatenate two bit string pads using a copy of this connector assembly, and in one more stage concatenate the third.
$\mathcal{O}(1)$ tile types, bins and stages can be used to ``fill in'' the portions of the assembly with width less than~7.

Steps $1$, $2$, and $4$ each use $\BO{\frac{\log \log b}{\log t}}$ stages, whereas Step $3$ uses $\BO{\frac{x-tb-t\log t}{b^2}+\frac{\log \log b}{\log t}}$ stages.
Then the total number of stages used by the target system is $\BO{\frac{x-tb-t\log t}{b^2}+\frac{\log \log b}{\log t}}$.
\end{proof}

The additive gap between the upper and lower bounds implied by these lemmas is due to the $\BO{\frac{\log \log{b}}{\log t}}$ additional stages used to assemble some of the machinery needed to carry out the three steps of Lemma~\ref{lem:xbit_string}.
We leave the removal of this gap as an open problem, stated in Section~\ref{sec:conclusion}.

%% file: bit_string_pad_construction.tex
\section{Bit String Pad Construction} \label{sec:bit_string_pad}

As sketched in Section~\ref{sec:key_lemmas}, we assemble bit string pads from three subpads constructed via distinct methods utilizing distinct sources of information complexity in staged self-assembly systems.
The construction is described in five parts constituting the remainder of this section.
The first two parts, Sections~\ref{sec:wings} and~\ref{sec:fattening}, describe helpful subconstructions used by the three subpad constructions that follow in Sections~\ref{sec:step_1} through~\ref{sec:step_3}.

\subsection{Wings} \label{sec:wings}
\input{wings.tex}

\subsection{Fattening} \label{sec:fattening}
\input{fattener.tex}

\subsection{Step 1: encoding via initial tile-to-bin assignment} \label{sec:step_1}
\input{step_1.tex}

\subsection{Step 2: encoding via tile types} \label{sec:step_2}
\input{step_2.tex}

\subsection{Step 3: encoding via mix graph} \label{sec:step_3}
\input{step_3.tex}

%% file: wings.tex

In this section, a $b$ bin system constructs all possible $\lfloor \log b \rfloor$-bit string pads, each contained in its own bin.
The purpose of these bit string pads is to mix them with $\BO{1}$ tile types to assemble \emph{wings}.
Wings are rectangular assemblies with geometric bumps, or \emph{teeth}, that encode a positive integer \emph{index} in binary.
A wing gadget has index $i$ and $m$ bits provided it geometrically encodes an $m$-bit binary string representing $i$.
Wings come in two varieties, \emph{west} and \emph{east}, shown in Figure~\ref{fig:wings}.

\begin{figure}[t]
    \centering
    \subcaptionbox{\label{fig:left_wing}}{\includegraphics[width=0.25\textwidth]{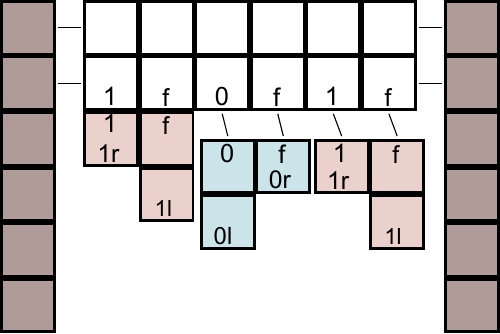}}
    \hspace*{.4cm}
    \subcaptionbox{\label{fig:east_wing}}{\includegraphics[width=0.25\textwidth]{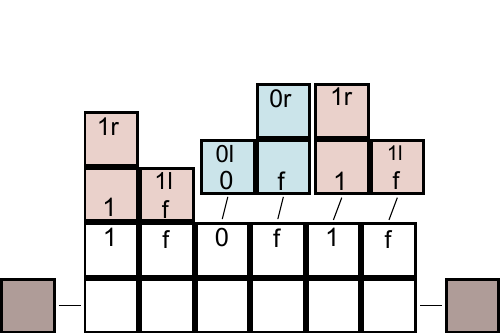}}
    \hspace*{.4cm}
    \subcaptionbox{\label{fig:bit_strip_wing}}{\includegraphics[width=0.42\textwidth]{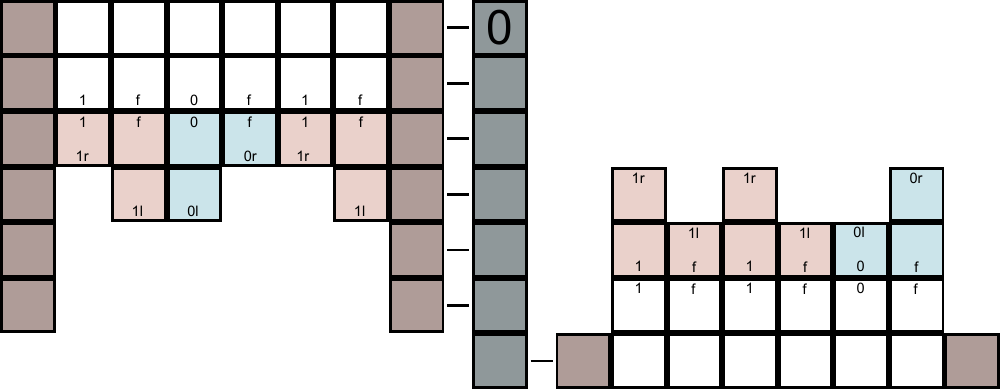}}
    \vspace*{\floatsep}
    \subcaptionbox{\label{fig:wing_example}}{\includegraphics[width=.8\textwidth]{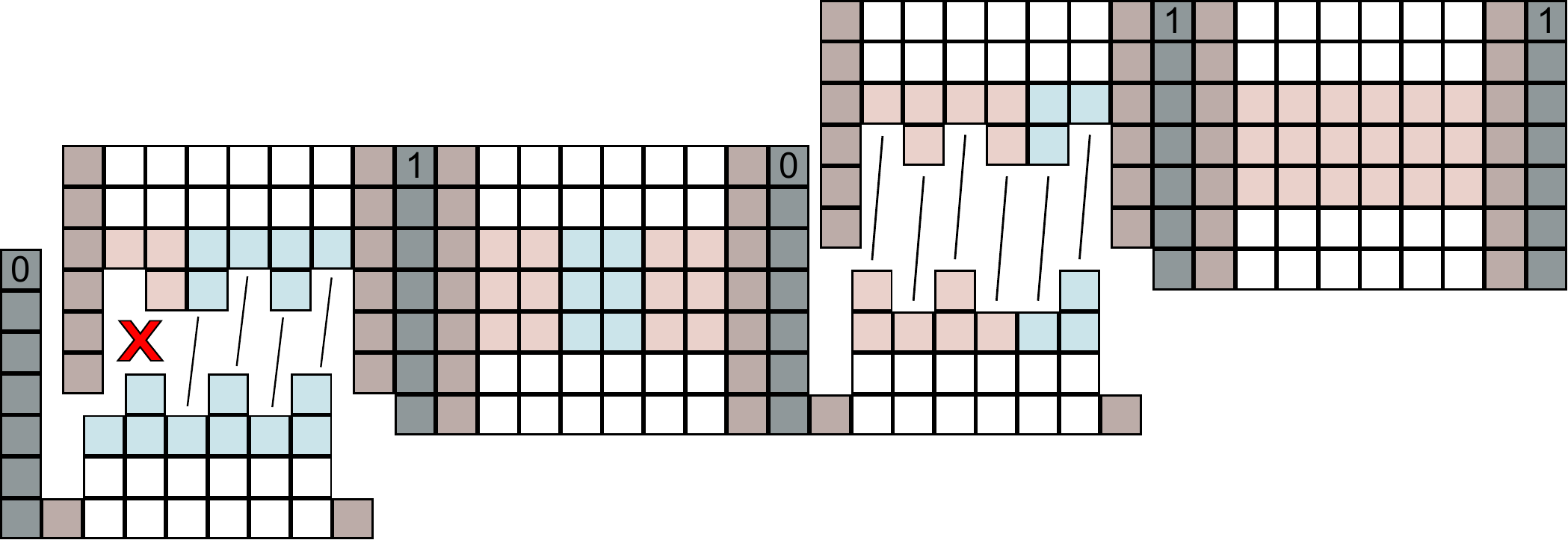}}
    \caption{Assembling wings from bit string pads and extra $\BO{1}$-size assemblies. (a) Assembly of a west wing. (b) Assembly of an east wing. (c) The attachment of chosen west and east wings to a $\BO{1}$-sized \emph{bit stick}. (d) Directing the assembly of bit sticks with attached wings. Left: mismatched wings prevented from attaching via matched glues due to mismatched geometry. Right: matching wings have no such geometric mismatch and so attach.}
\label{fig:wings}
\end{figure}

\begin{figure}[t]
   \centering
    \includegraphics[width=.65\textwidth]{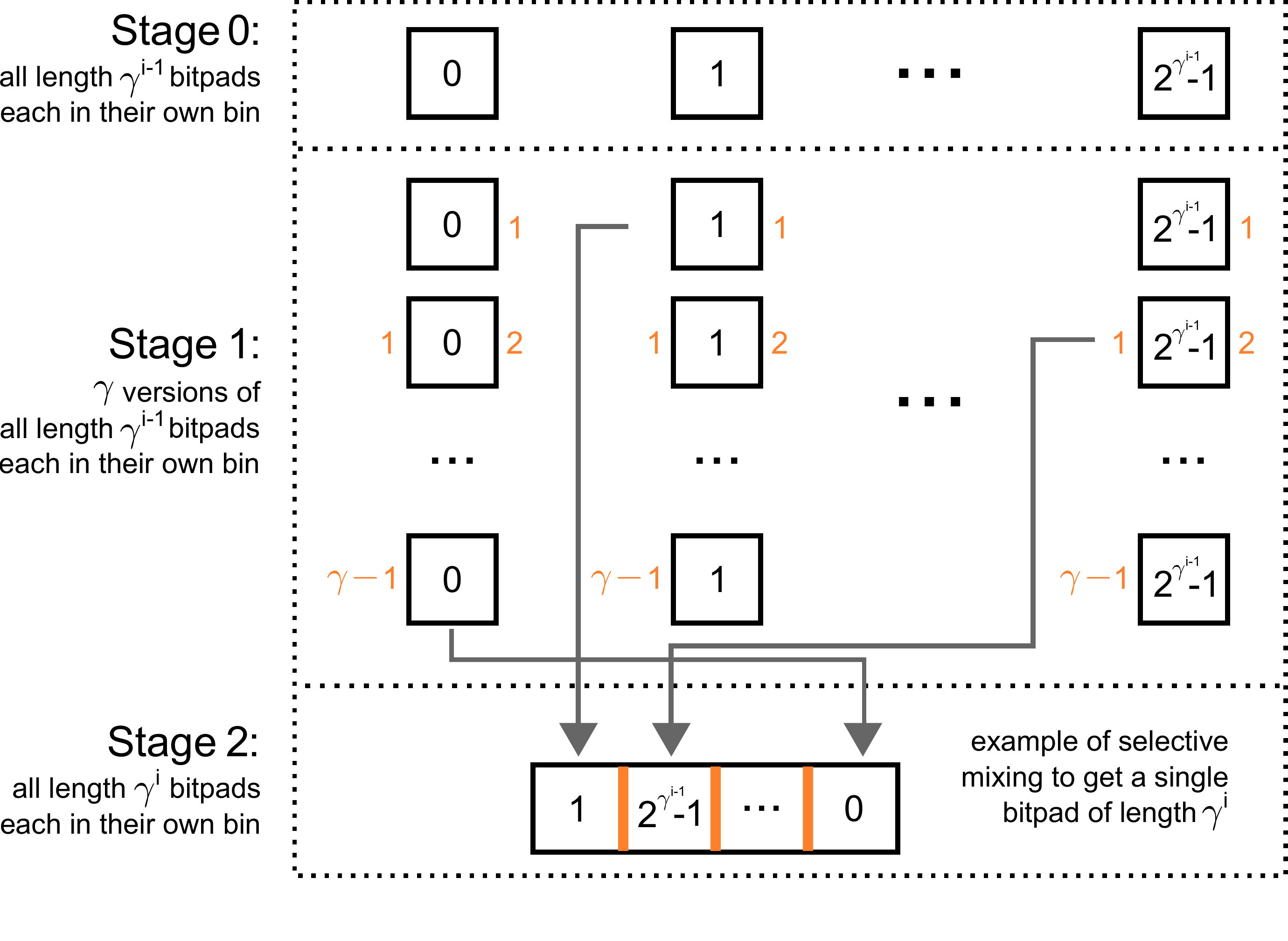}
    \caption{The 2-stage Round~$i$ used in Case~1 of Lemma~\ref{lem:wings}.} 
    \label{fig:wing_building}
\end{figure}

Wings encode their indices on either the north surface (west wings) or south surface (east wings).
Roughly speaking, a pair of west and east wings can attach if their teeth match (perfectly interlock), equivalent to having the same index.
In Steps~1 and~3, west and east wings are attached to the west and east ends of \emph{bit sticks}, which are $7\times1$ assemblies which expose a north-facing $g_0$ or $g_1$ glue.
The attached wings direct these bit sticks to assemble in a desired west-to-east order as seen in Figure~\ref{fig:wing_example}.
The bit string pads underlying the wings are constructed as follows:

\begin{lemma}
\label{lem:wings}
There exists a constant $c$ such that for any $b,t \in \N$ with $b,t > c$, there exists a $\tau = 2 $ staged assembly system with $b$ bins, $t$ tile types, and $\BO{\frac{\log{\log{b}}}{\log{t}}}$ stages whose uniquely produced output is all width-2 gap-1 $\lfloor\log(b)\rfloor$-bit string pads, each placed in a distinct bin.
\end{lemma}

\begin{proof}
Let $\gamma = \lfloor \frac{t-4}{2} \rfloor + 1$.
The construction has three cases, depending upon the relative values of $\gamma$ (determined by $t$) and $b$.

\textbf{Case 1: $\gamma < 2 \log{b}$ and $b=\gamma^n$ for $n \in \mathbb{N}$.}
All three cases use width-2, gap-0 1-bit string pads called \emph{binary gadgets}: $\BO{1}$-sized assemblies representing~0 and~1.
The remaining tile types are used to assemble $2(\gamma - 1)$ \emph{connectors} that bind to the west and east ends of binary gadgets.
Connectors come in \emph{west} and \emph{east} varieties and connector has a positive integer \emph{index} such that two connectors can bind to one another if and only if their indices are equal.

For each $k \in \{1,2,\dots,\gamma-1\}$, assemble west and east connectors with index $k$.
Then repeat the following 2-stage ``round'' seen in Figure~\ref{fig:wing_building}.
The $i$th round begins with $2^{\gamma^{i-1}}$ binary gadgets, each in their own bin, and each of the west and east connectors stored in their own bins.
In the first stage of the round, mix each binary gadget with west and east connectors with indices $h-1$ and $h$, respectively, for all $h \in \{1,2,\dots,\gamma-1\}$,
In the special case that $h = 1$ or $h = \gamma-1$, omit the west or east connector, respectively.
The result is $\gamma2^{\gamma^i}$ distinct assemblies, with each assembly in its own bin.
The tile set for bit string pads and connectors can be seen in Figure~\ref{fig:wing_tiles}.

\begin{figure}[t]
    \centering
    \includegraphics[width=.8\textwidth]{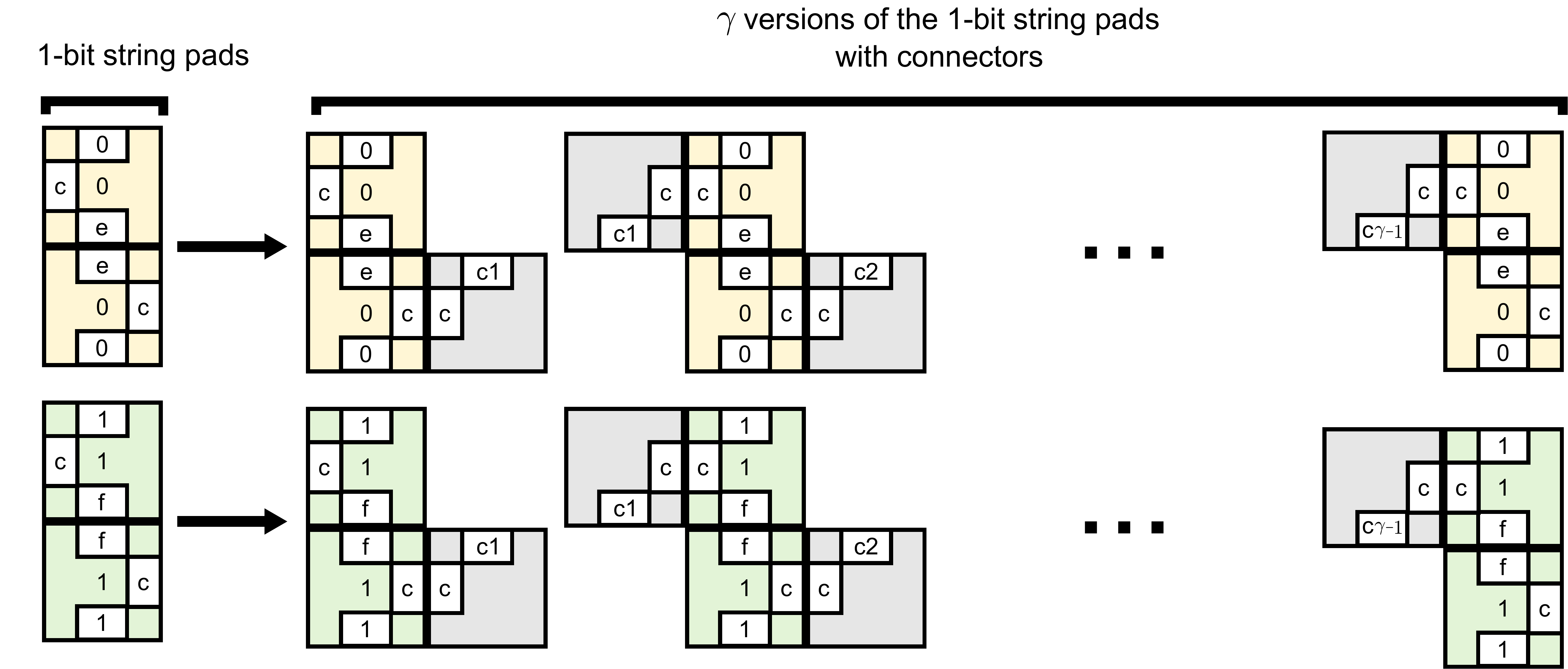}
    \caption{Left: 2~width-2 1-bit string pads. Right: $2\gamma$ versions of each bit string pad with pairs of west and east connectors attached, as done in the first stage of each round.}
    \label{fig:wing_tiles}
\end{figure}

In the second stage of the round, selectively mix the $\gamma$-size subsets of the $\gamma2^{\gamma^{i-1}}$ bins, choosing one binary gadget with each pair of connector indices, to assemble all width-2 $\gamma^i$-bit string pads in separate bins.
The number of bins used in round $i$ is thus $2^{\gamma^i} + \gamma$: enough for all $\gamma^i$-bit string pads and the $\gamma$ (reusable) connectors.

Perform sufficient rounds to assemble all $\lfloor \log(b) \rfloor$-bit string pads, i.e., the smallest integer $r$ such that $\gamma^r \geq \log(b)$ and thus $r = \lfloor \log_\gamma{\log(b)} \rfloor$.
In the last stage of round $r$, the number of bins used can be reduced to $2^{\gamma^r}$, dropping the $\gamma$ additional bins containing connectors.

The number of bins used then does not exceed $b$ since $2^{\gamma^i} + \gamma \leq 2^{\gamma^r} \leq b$ for all $i < r$.
Moreover, the number of stages used is $\BO{r} = \BO{\frac{\log\log{b}}{\log{\gamma}}} = \BO{\frac{{\log\log{b}}}{\log{t}}}$.
This construction requires $b,t > c$ for some constant $c$ large enough to ensure $\gamma \geq 1$ and $b$ can accommodate at least $1$ of each of the bin types.

\textbf{Case 2: $\gamma < 2 \log(b)$ and $b \neq \gamma^n$ for $n \in \mathbb{N}$.}
Note that this case is identical to Case~1, except that $b$ is not a power of $\gamma$.
Thus the desired length of the assembled bit string pads are between $\gamma^{r-1}$ and $\gamma^r$ for some $r$; as a solution, bit string pads are assembled from collections of shorter bit string pads of power-of-$\gamma$ lengths.

Let $d_{\lfloor \log_\gamma(b) \rfloor} d_{\lfloor \log_\gamma(b) \rfloor - 1} \dots d_2 d_1$ be the base-$\gamma$ expansion of $b$.
New \emph{special} connectors with index~0 are used, in addition to the \emph{standard} connectors from Case~1.
Each round also uses \emph{growth} bins, \emph{output} bins, and \emph{west} and \emph{east incubator} bins.
West and east incubator bins contain growing $\sum_{j=1}^i{d_j\gamma^{j-1}}$-bit string pads with special west or east connectors, respectively.

As in Case~1, use $\log_\gamma(b)$ rounds to assemble growing sets of longer bit string pads.
The $i$th round begins with all $\gamma^{i-1}$-bit string pads in separate growth bins.
In the first stage of the $i$th round, mix standard connectors with these pads to assemble $\gamma$ versions of all $\gamma^{i-1}$-bit string pads in separate growth bins.
Also, the output bins are equal to their predecessors in the previous stage.

In the second stage of the $i$th round, if $d_i > 0$, then selectively mix $\gamma$-sized subsets of bit string pads with standard connectors (stored in growth bins) to assemble all $d_i\gamma^{i-1}$-bit string pads in separate bins.
If $i = 1$ or the output bins are empty, these bit string pads are stored in their own output bins.
Otherwise, the output bins are not empty; in this case they are also mixed with west connectors and stored in west incubator bins.
In the same stage, mix every bit string pad from the output bins with a special east connector and store them in east incubator bins.

In the third stage, mix all pairs of west and east incubator bins into separate output bins, replacing any previous output bins.

\begin{figure}[t]
   \centering
    \includegraphics[width=.9\textwidth]{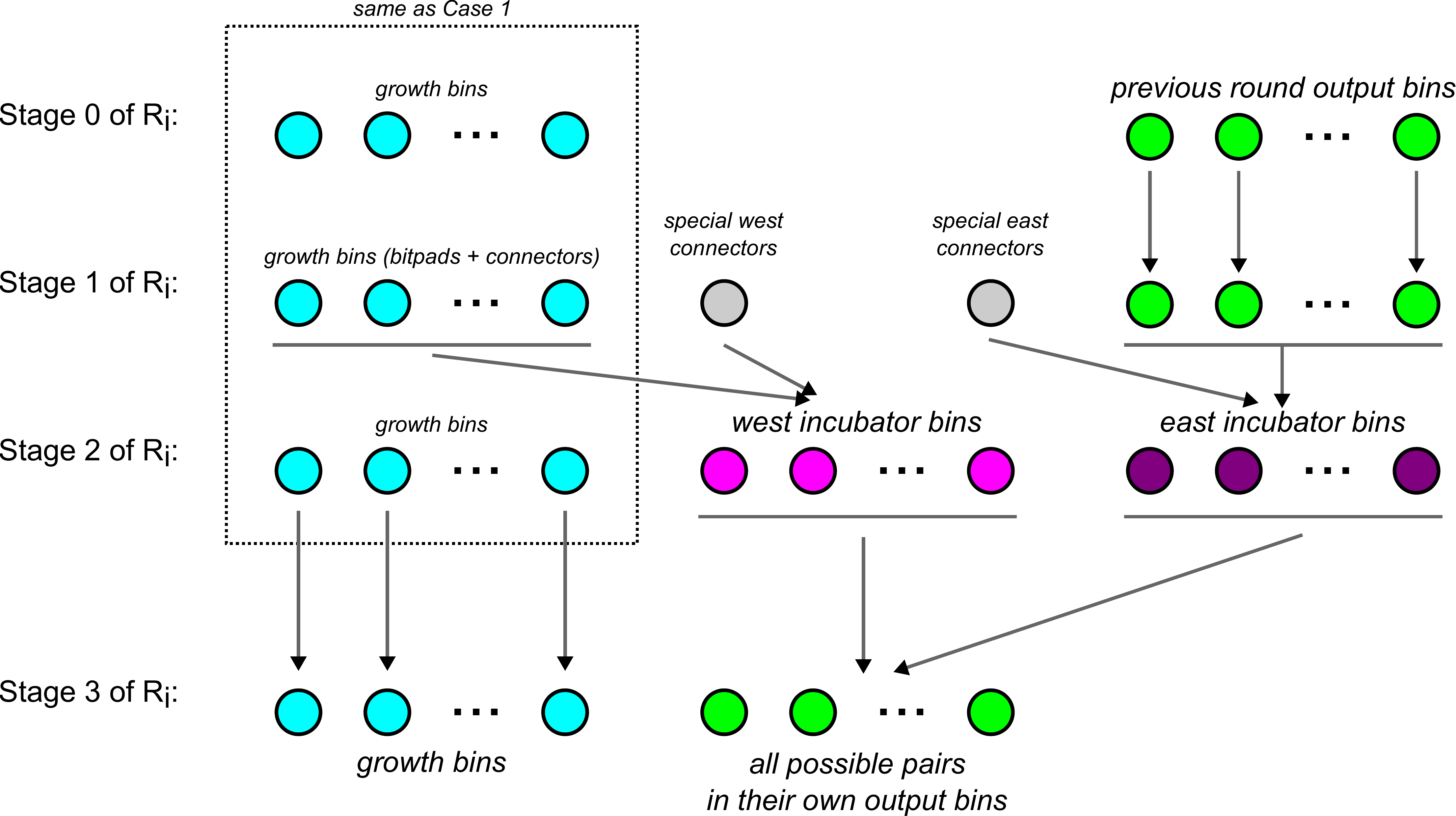}
    \caption{The 3-stage Round~$i$ used in Case~2 of Lemma~\ref{lem:wings}.}
    \label{fig:case_2_highlevel}
\end{figure}

For all but the last round, mix $\gamma$-sized subsets of bit string pads with standard connectors (stored in growth bins) to assemble all $\gamma^i$-bit string pads in separate growth bins, replacing the previous growth bins.
Since each round concatenates $d_i\gamma^{i-1}$ bits to the bit string pads in output bins, the bit string pads assembled in the output bins of the final round represent $\sum_{j=1}^{\lfloor \log_\gamma(b) \rfloor}{d_j\gamma^{j-1}} = \lfloor \log(b) \rfloor$ bits.
The high-level idea for the construction can be seen in Figure \ref{fig:case_2_highlevel}.

This construction requires $b, t > c$ for some constant $c$ large enough to ensure that $\gamma \geq 1$ and at least one of each bin type (growth, output, etc.) is available.
The number of stages remains $\BO{\frac{\log{\log{b}}}{\log{t}}}$.

\textbf{Case 3: $\gamma \geq 2 \log(b)$.}
Let $\beta = \lfloor \log(b) \rfloor$.
Begin with $\beta$ bins, each containing a binary gadget representing~0 with a distinct pair of west/east glues $g_1$ and $g_2$, $g_2$ and $g_3$, etc.
Similarly, create a second set of identical $\beta$ bins, but with binary gadgets representing~1.
In the second stage, combine every $\beta$-sized subset of bins that contain binary gadgets with distinct west/east glue pairs to assemble all $\beta$-bit string pads.
\end{proof}

%% file: fattener.tex

Roughly speaking, ``fattening'' replaces a single tile type with a set of~5 corresponding tile types and a generic \emph{filler} assembly of length~$k-2$ to yield a $2\times k$ assembly which exposes the same glues as the original tile (see Figure~\ref{fig:fattener}).
Fattening is used as a subroutine in Steps~1 and~2 to increase the gap of the subpads they assemble from 0 to $\Theta(\log{b})$, matching that of the subpad assembled by Step~3, and as a subroutine in Step~2 for assembling decompression pads.

A \emph{filler assembly} is an assembly of length $k-2$ used to fatten a tile type to some length $k$.
The technique in Section~\ref{sec:wings} for assembling wings can be used to assemble length-$\Theta(\log{b})$ filler assemblies.
Lemma~\ref{lem:wings} yields all length-$\log(b)$ bit string pads in $\BO{\frac{\log\log{b}}{\log{t}}}$ stages; small modifications yields a filler assembly of the same length in the same complexity.

\begin{figure}[t]
\centering
\subcaptionbox{}{\label{fig:fattener_tile}\includegraphics[width=0.12\textwidth]{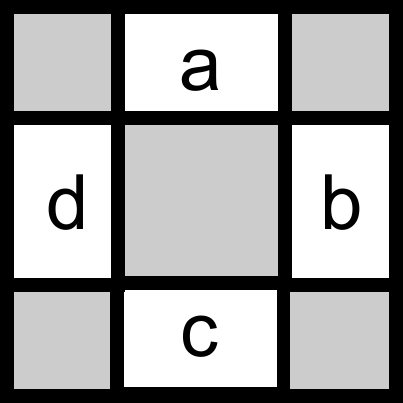}}
\hspace{1cm}
\subcaptionbox{}{\label{fig:fattener_mix}\includegraphics[width=0.6\textwidth]{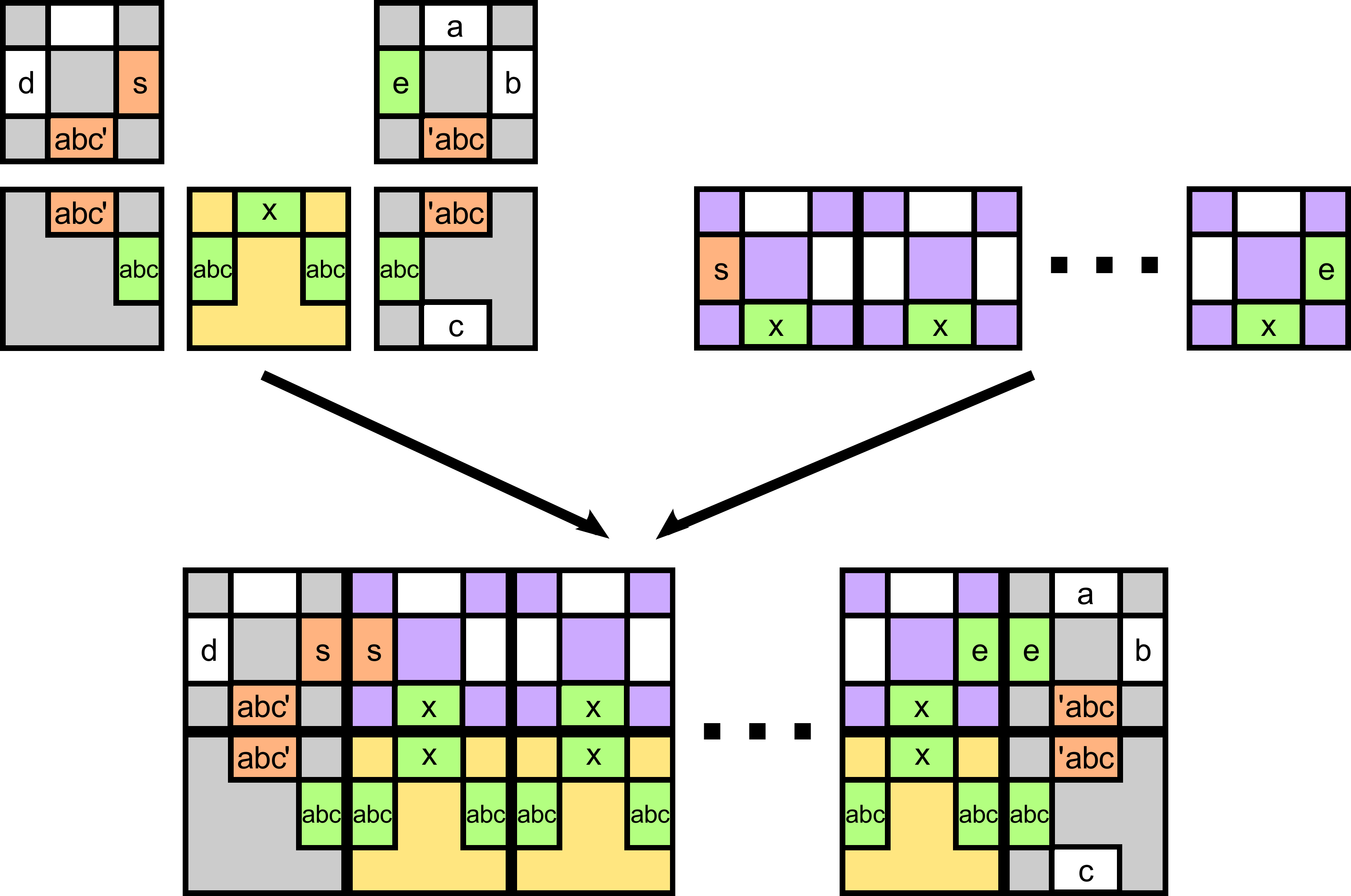}}
\caption{(a) The tile to be fattened. (b) 5 tile types, based on the tile to be fattened, are mixed with a length-$(k-2)$ filler assembly which ``fattens'' the assembly from width $1$ to width $k$. The four glues of the original are exposed in the same directions.}
\label{fig:fattener}
\vspace{-2mm}
\end{figure}

%% file: step_1.tex
\begin{figure}[t]
    \centering
    \includegraphics[width=0.85\textwidth]{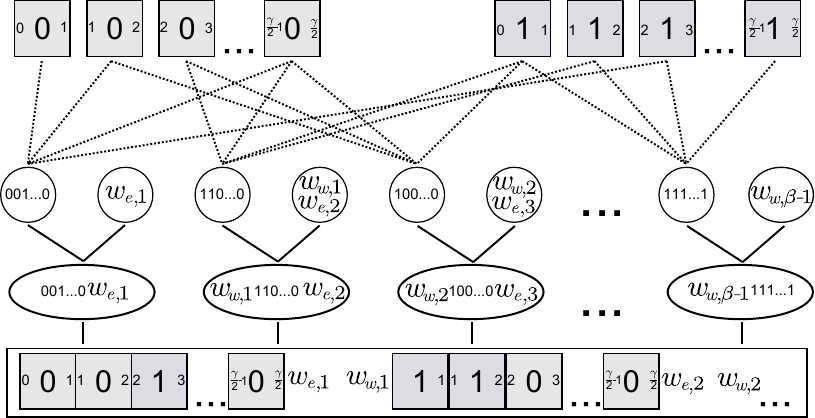}
    \caption{The assembly of a $\frac{\gamma\beta}{2}$-bit string pad. The squares labeled $0$ and $1$ represent bit sticks. The dotted lines indicate tile to bin assignments before the first stage of the system; $w_{e,i}$ and $w_{w,i}$ represent the $i$th east and west wings respectively.}
    \label{fig:step1_abstract}
\end{figure}

Recall that in a staged system, each of the system's $b$ stage-1 bins is assigned a subset of $t$ total tile types.
Here we design an assignment that assembles a $\Theta(tb)$-bit string subpad of the final bit string pad using $\BO{\log{\log{b}}/\log{t}}$ stages - dominated by the use of the wings subconstruction of Section~\ref{sec:wings}.
The assignment yields $\Theta(b)$ bins that contain subpads encoding distinct length-$\Theta(t)$ substrings of the $\Theta(tb)$ bits.
These assemblies are then combined (in the proper order) using wings.

\begin{lemma}
\label{lem:step_1}
There exists a constant $c$ such that for any $b,t \in \N$ with $b,t > c$, there exists an $x = \Theta(tb)$ such that for any bit string $S$ of length $x$, there exists a $\tau = 2$ staged assembly system with $b$ bins, $t$ tile types, and $\BO{\frac{\log\log b}{\log t}}$ stages whose uniquely produced output is a width-$7$ gap-$\Theta(\log b)$ $x$-bit string pad representing $S$.
\end{lemma}

\begin{proof}
See Figure~\ref{fig:step1_abstract} for a sketch of the idea.
Let $\gamma$ and $\beta$ be constant fractions of $t$ and $b$, respectively, and $b, t$ large enough such that $\gamma, \beta \geq 1$.
Use~$\gamma$ tiles and~$\beta$ bins to construct all west and east $\log(\beta)$-bit wings according to Section~\ref{sec:wings}.
Also construct $\frac{\gamma}{2}$ distinct \emph{bit sticks}.
Bit sticks are width-$7$ gap-$0$ $1$-bit string pads, i.e. $7 \times 1$ assemblies that expose $g_1$ or $g_0$ glues to the north.
Additionally, the bit sticks' east and west sides expose glues such that any $\frac{\gamma}{2}$-bit string pad can be assembled from $\frac{\gamma}{2}$ bit sticks attached sequentially.

Each bit string pad constructed this way must also be \emph{fattened} using the technique described in Section~\ref{sec:fattening}.
Using $\BO{\frac{\log \log \beta}{\log \gamma}}$ stages, each opads is fattened to length $g$ where $g$ is the length of a wing.
The purpose of fattening is to ensure the complete bit string pad constructed has uniform gap $g$.

In each of~$\beta$ bins, assemble $\frac{\gamma}{2}$ bit sticks into a distinct $\frac{\gamma}{2}$-bit string subpad of the desired pad.
Combine each of these~$\beta$ length-$\frac{\gamma}{2}$ subpads with wings that encode their sequential ordering in the pad (i.e. the wings are used to ensure that the subpads attach in the desired order).
Then, assemble these ``wing-labeled'' subpads into the complete $\frac{\gamma\beta}{2} = \Theta(tb)$-bit string pad.
The resultant bit string pad will have gap-$g$, since each bit in each subpad has been fattened to length $g$ and wings of length $g$ are used to concatenate each subpad.

This gap satisfies $g = \Theta(\log{b})$, since the wings used are 1-gapped $\log\beta$-bit string pads with $\BO{1}$ additional length added (as seen in Figure~\ref{fig:wings}).
Particularly, $g = 2(\log \beta) + 2$.
The number of stages used is $\BO{\frac{\log\log \beta}{\log \gamma}} = \BO{\frac{\log\log b}{\log t}}$ (for the wings, see Lemma~\ref{lem:wings}) plus $\BO{1}$ (the subpads of the desired pad).
\end{proof}

%% file: step_2.tex
Here the goal is to design a collection of $t$ tile types that assembles a target string of~$\Theta(t\log{t})$ bits in bit string pad form.
The solution is to utilize a common base conversion approach in tile assembly used by~\cite{AdChGoHu01,AGKS05g,FMS2015OPS,SolWin07}.
In this approach, tile types optimally encode integer values in a large base and are then ``decompressed'' into a binary representation.
In total, $t$ tile types are used to encode a value in a high base and decompress this representation into a string of $\Theta(t\log{t})$ bits.

\begin{definition}[Decompression pad]
For $k, d, x \in \mathbb{N}$ and $u=2^x$, a width-$k$ $d$-digit base-$u$ decompression pad is a $k\times dx$ rectangular assembly with $d$ glues from a set of $u-1$ glue types $\{g_0 , g_1,..., g_{u-1}\}$ exposed on the north face of the rectangle at intervals of length $x-1$ and starting from the westmost northern edge.
All remaining glues on the north surface have a common type $g_N$.
The remaining exposed south, east, and west tile edges have glues $g_S$, $g_E$, and $g_W$.
The exposed glues on the northern edge, disregarding glues of type $g_N$, $g_{a_1}, g_{a_2},\dots, g_{a_d}$ represent string $a_1, a_2,\dots, a_d$.
\end{definition}

The base conversion approach is illustrated by an example: assembling a width-3 9-bit string pad representing $S=010100000$ ($S=240$ in base~8).
First, a decompression pad representing~$S$ in base~8 by combining~3 different $3 \times \log_2(8)$ blocks is assembled (see Figure~\ref{fig:mt_pad_building}).
Then, this decompression pad is combined with $\BO{u}$ tile types to ``decompress'' the 3-digit base-8 representation of $S$ into an 8-bit base-2 representation (see Figure~\ref{fig:decompression_example}).

The remainder of this section consists of three lemmas.
The first, Lemma~\ref{lem:decompression_pad}, establishes that decompression pads representing $d$-digit base-$u$ strings can be assembled with $\BO{d + \log(u)}$ tile types.
The second, Lemma~\ref{lem:step_2_1}, describes the decompression tile types used to ``decompress'' the high-base representation on the decompression pad.
The third, Lemma~\ref{lem:step_2_2}, optimizes the choice of base used to minimize the number of tile types used in total for both the decompression pad and decompression tiles. 

The resulting bit string pads are gap-0; the follow-up Section~\ref{sec:gapped_step_2} describes how to increase the gap of these bit string pads to match those assembled in Steps~1 and~3 by applying fattening to a subset of the tile types used.

\begin{lemma}
\label{lem:decompression_pad}
Given integers $x \ge 3$, $d\ge 1$ and $u=2^x$, there exists a 1-stage, 1-bin staged $\tau = 2$ self-assembly system whose uniquely produced output is a $d$-digit decompression pad of width-2 and base-$u$, using at most $5d+\log(u)-2$ tile types.
\end{lemma}

\begin{proof}
Consider a number $S = s_0 s_1 \dots s_{d-1}$ in a base $u$.
The goal is to build a \emph{decompression pad} with north glues representing the $d$ digits of $S$, and $\log(u)-1$ spacing between consecutive digits exposing north glues $g_B$.
As depicted in Figure~\ref{fig:mt_pad_building}, a width-2 length-$\log(u)$ decompression pad is built for each of the $d$ digits in the base-$u$ representation of $S$.
These $d$ length-$\log(u)$ decompression pads then assemble into the complete length-$d\log(u)$ decompression pad representing $S$. 

\begin{figure}[t]
    \centering
    \includegraphics[width=0.49\textwidth]{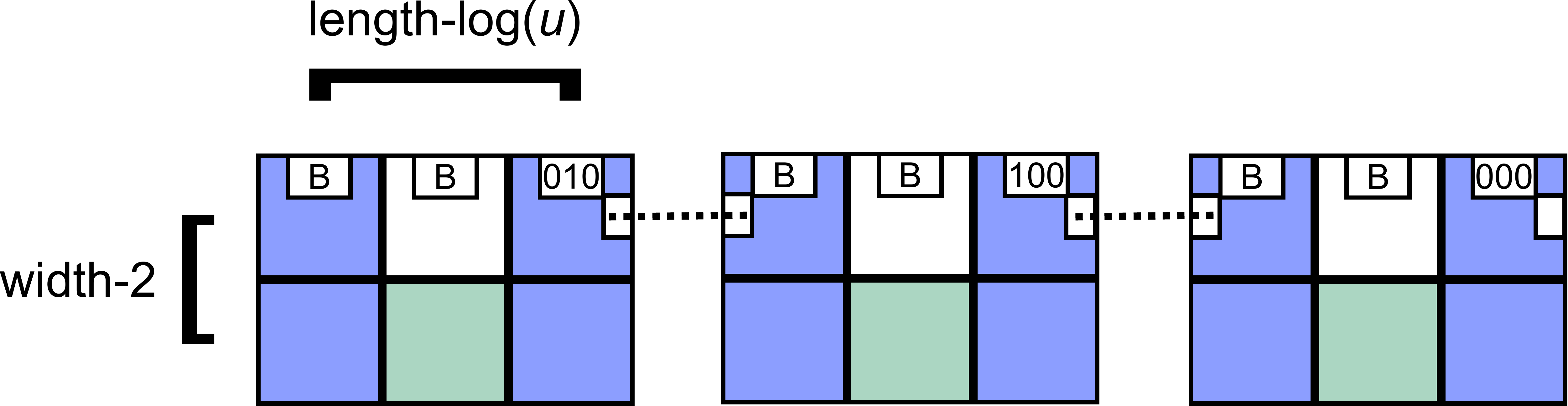}
    \includegraphics[width=0.49\textwidth]{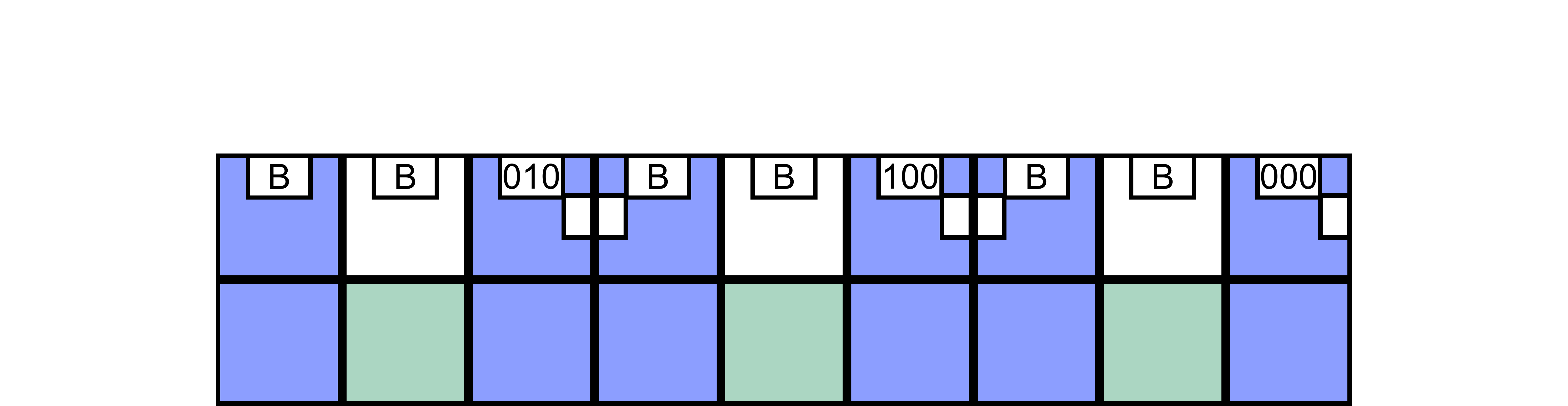}
    \caption{Assembling a decompression pad representing a bit string $S=010\,100\,000$ in base $u = 8$. 
Left: three width-2 length-$\log(u)$ decompression pads are designed for each digit in the base-$u$ representation of $S$. 
Right: the decompression pad assembled into a width-2 length-$|S|$ decompression pad representing $S$ in base $u$.}
    \label{fig:mt_pad_building}
\end{figure}

The decompression pads for each digit require~5 unique tile types along with a \emph{shared filler} of length $\log(u)-2$ built with $\log(u)-2$ tile types. The shared filler is an optimization that allows a portion of each length-$\log(u)$ decompression pad to be assembled with some shared tile types.
Therefore, the total tile complexity is $5d + \log(u)-2$.
\end{proof}

\begin{lemma}
\label{lem:step_2_1}
Given integers $d\ge 3$, $x \ge 3$, $u=2^x$, and a $d\log(u)$-bit number $S$, there exists a $\tau = 2$ staged assembly system with 1~bin, $5d+2u+\log(u)-4$ tile types, and 1~stage that uniquely produces a width-3 gap-0 $d\log(u)$-bit string pad representing $S$.
\end{lemma}

\begin{proof}
Since $S$ is a $d\log(u)$-bit number, it can be represented as a $d$-digit number in base $u$. 
Use Lemma~\ref{lem:decompression_pad} to obtain a tile set that assembles into a decompression pad that represents $S$ in base $u$ (requiring $d$ digits and thus $5d+\log(u)-2$ tile types).

\begin{figure}[t]
    \centering
    \includegraphics[width=.85\textwidth]{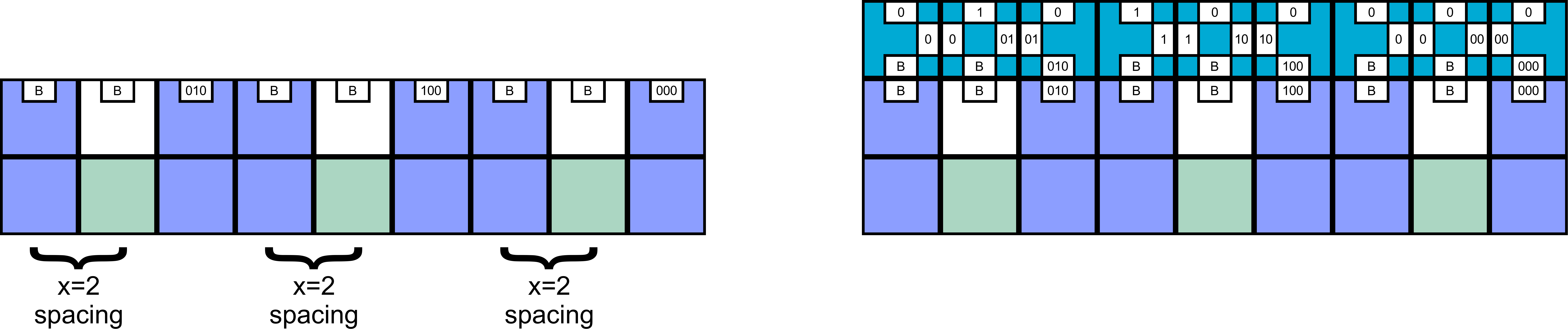}
    \caption{
Left: a width-2 gap-$(\log(u)-1)$ decompression pad representing a bit string $S=010100000$ in base $u = 8$.
Right: $\BO{u}$ decompression tiles interact with the north glues of the decompression pad to combine into a width-3 bit string pad representing $S$ in base~2.}
    \label{fig:decompression_example}
\end{figure}

Use an additional $2u-2$ tile types ``decompress'' the high-base representation of $S$ by exposing a single bit via north glues and passing along the remaining bits of the digit via east and west glues (see Figure~\ref{fig:mt_decompression_tiles}). 

\begin{figure}[t]
    \centering
    \includegraphics[width=1\textwidth]{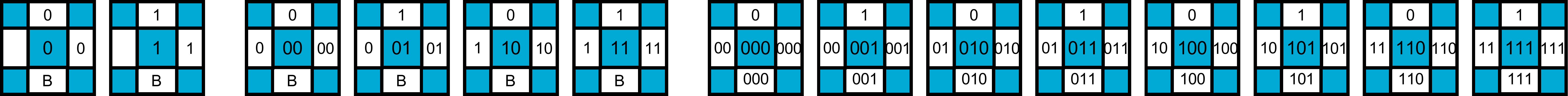}
    \caption{The set of tiles used to decompress a decompression pad representing a string in base~8. Since $u=8$, $2u-2 = 14$ tile types are used. These tiles attach to the exposed glues on a decompression pad, replacing a large-gap high-base representation of $S$ with a 0-gap low-base representation.}
    \label{fig:mt_decompression_tiles}
\end{figure}

In total, $\log(u)$ attachments are used to unpack the base-$u$ digit into a sequence of $\log(u)$ bits. 
The total tile complexity is $t = 5d+2u+\log(u)-4$: $5d+\log(u)-2$ types to build the decompression pad and $2u-2$ decompression tile types.
\end{proof}

\begin{lemma}
\label{lem:step_2_2}
There exists a constant c such that for any $t \in \mathbb{N}$ with $t > c$, there exists an $x=\theta(t\log t)$ such that for any bit string $S$ of length $x$, there exists a $\tau = 2$ staged assembly system with 1~bin, $t$~tile types, and 1~stage that uniquely assembles a width-$3$ gap-$0$ $x$-bit string pad representing $S$.
\end{lemma}

\begin{proof}
Here, the large base $u$ used Lemma~\ref{lem:step_2_1} is optimized to maximize the number of bits encoded by the construction.
Recall that Lemma~\ref{lem:step_2_1} first constructs length-$\log(u)$ decompression pads for every digit in a base-$u$ string.
Decompression pads are constructed using~5 tile types plus some additional number of shared ``spacing'' tile types that determine the gap length of the decompression pad.
Then ignoring spacing tile types, $\lfloor \frac{t}{2} \rfloor$ tile types assemble $\lfloor \frac{t}{10} \rfloor$ decompression pads and so $\lfloor \frac{t}{10} \rfloor$ base-$u$ digits can be represented.

Next, Lemma~\ref{lem:step_2_1} uses $2u + \log(u)$ tile types for ``decompressing'' the decompression pads and shared spacing.
Since $\lfloor \frac{t}{2} \rfloor$ tile types were used to assemble decompression pads, only $\lfloor \frac{t}{2} \rfloor$ available tile types remain and thus $2u + \log(u) \leq \lfloor \frac{t}{2} \rfloor$.

The number of bits $x$ encoded by the assembled bit string pad is $x = d\log(u)$.
The choice of $u = 2^{\lfloor \log{\frac{t}{3}} \rfloor}$ satisfies the aforementioned bound on tile types while (asymptotically) maximizing $u$. 
The result is a bit string pad with $x = d\log(u) = \left\lfloor \frac{t}{10} \right\rfloor \log{2^{\lfloor \log {\frac{t}{3}} \rfloor}} = \left \lfloor \frac{t}{10} \right\rfloor \left\lfloor \log{\frac{t}{3}} \right\rfloor = \Theta{(t \log t)}$ bits.

Due to constraints building the decompression pad, the high-base string must have at least~3 digits and base at least~8.
For building the decompression pad, minus the shared spacing, at least $\lfloor \frac{t}{2} \rfloor \ge 5d=15$ tile types are needed.
For the decompression tiles and shared spacing, at least $\lfloor \frac{t}{2} \rfloor \ge 2u +\log(u)-4 = 15$ tile types are needed.
For base~8, $u=2^{\lfloor \log{\frac{t}{3}} \rfloor}=8$. 
So this optimization requires $\frac{t}{2} \geq 15$ and $\frac{t}{3}\ge 8$ and thus $t \ge c=30$.
\end{proof}

\subsubsection{Gapped bit string pads}\label{sec:gapped_step_2}
The subpads assembled by Section~\ref{sec:step_1} and~\ref{sec:step_3} are gap-$\Theta(\log{b})$, in contrast with the gap-0 subpad assembled here.
As a solution, we increase the gap of these subpads from~0 to $\Theta(\log{b})$ using a modified version of the fattening process described in Section~\ref{sec:fattening}.

Let $q = \Theta(\log{b})$ be the length of the desired gap.
In Lemma~\ref{lem:step_2_2}, width-2 decompression pads representing strings of base-$u$ digits are assembled.
To increase the gap, a portion of the tile types used to assemble these decompression pads are replaced with fattened versions to space out the $g_B$ glues with $g_f$ glues, resulting in a width-3 decompression pad. Figure~\ref{fig:fattening_decompression_pad} shows an example of the modification.

\begin{figure}[ht]
    \centering
    \includegraphics[width=.7\textwidth]{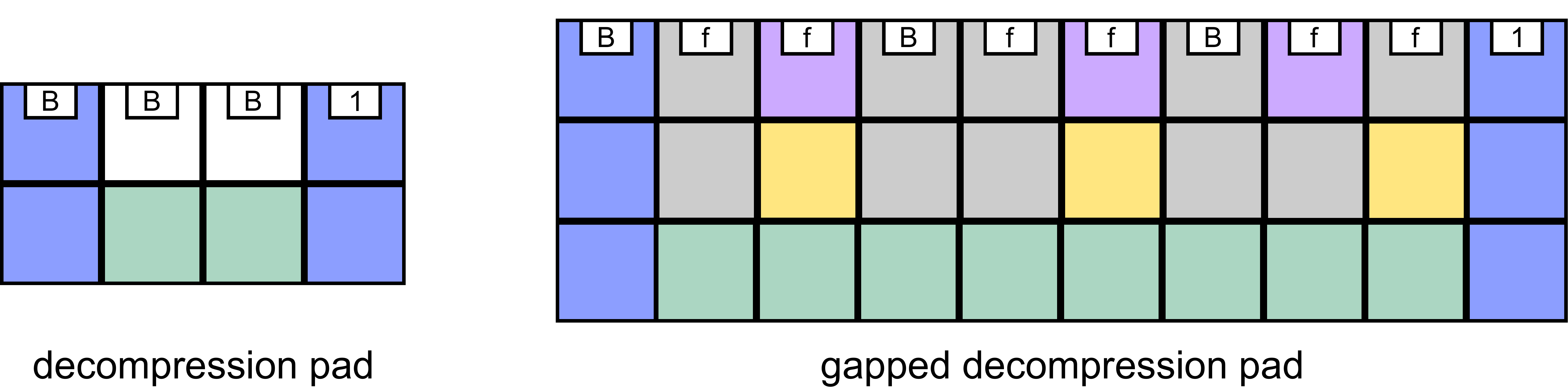}
    \caption{Gapped decompression pads: Here we see a gap-0 width-2 decompression pad (left) and a gap-2 width-3 decompression pad (right). We only use a constant more tile types per decompression pad and some slight modifications to the glues.}
\label{fig:fattening_decompression_pad}
\end{figure}

%% file: step_3.tex
In this step, a system is designed wherein the mix graph is carefully designed to choose bits by mixing particular bins, a technique termed \emph{crazy mixing}~\cite{DDFIRSS07}. 
Observe that there are $b^2$ potential mix graph edges between a pair of adjacent stages of a $b$-bin staged assembly system.
Thus any choice of edges can be specified using $b^2$ bits - one bit per edge.

\begin{figure}[t]
        \centering
        \includegraphics[width=1.0\textwidth]{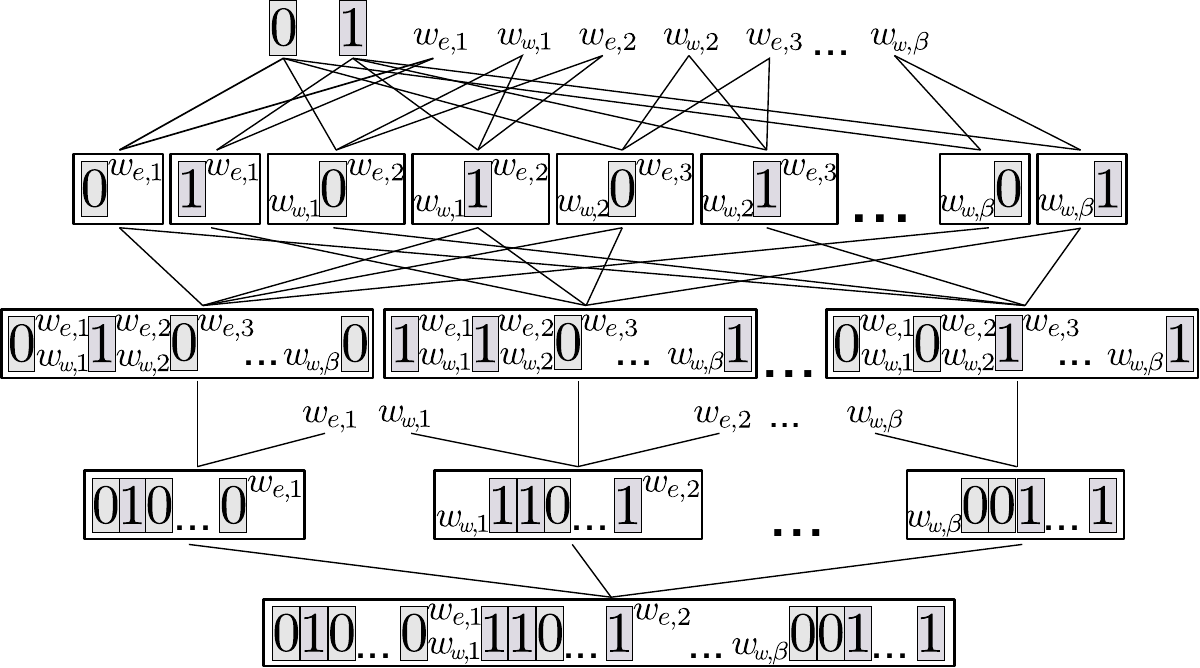}
        \caption{The creation of $\beta^2$-bit string pads using $\beta$ wings and $\BO{1}$ stages. The rectangles~0 and~1 represent bit sticks that may attach wings on either side; $w_{e,i}$ and $w_{w,i}$ represent the $i$th east and west wings respectively.}
        \label{fig:step3_abstract}
\end{figure}

\begin{lemma}
\label{lem:step_3}
There exists a constant $c$ such that for any $b,t\in \mathbb{N}$ with $b,t > c$ and any bit string $S$ of length $x$, there is a $\tau = 2$ staged assembly system with $b$ bins, $t$ tile types, and $\BO{\frac{x}{b^{2}} + \frac{\log \log b}{\log t}}$ stages that unique assembles a width-$7$ gap-$\Theta(\log{b})$ $x$-bit string pad representing~$S$.
\end{lemma}

\begin{proof}
The construction divides the target $x$-bit string pad into $b^2$-bit subpads, and a $\BO{1}$-stage process is repeated to specify the $b^2$ bits of each of these subpad.
An overview of the construction is shown in Figure~\ref{fig:step3_abstract}.

Let $\gamma$ and $\beta$ represent some constant fractions of $t$ and $b$ respectively.
Utilize $\gamma$ tiles and $\beta$ bins to construct all length-$\log_2(\beta)$ west and east wings according to Section~\ref{sec:wings}.
Denote the $i$th west and east wings by $w_{w,i}$ and $w_{e, i}$, respectively.
Also construct two \emph{bit sticks}: width-7, gap-0 1-bit sting pads, i.e. $7\times1$ assemblies that expose~1 or~0 north glues.

\paragraph{Stage~1.}
Mix $w_{e,i}$ and $w_{w,i-1}$ with bit stick~0 into a bin labeled $b^0_i$ for all $1 \leq i \leq \beta$.
Similarly, mix $w_{e,i}$ and $w_{w,i-1}$ with bit stick~1 into a bin labeled $b^1_i$.
In total, $2\beta$ bins are used.

\paragraph{Stage~2.}
Selectively mix specific bit sticks to assemble specific $\beta$-bit string pads across $\beta$ bins.
Specifically, mix either $b^0_i$ or $b^1_i$ for each $i$ across $\beta$ bins for a total of $\beta$ different $\beta$-bit string pads.

\paragraph{Stage~3.}
Attach wings to each of the $\beta$-bit string pads.
For each of the $\beta$ bins, mix $w_{e,i}$ and $w_{w,i-1}$ into the bins such that $w_{e,1}$ is mixed with the first $\beta$ bits of the desired $\beta^2$-bit string pad, $w_{e,2}$ and $w_{w,1}$ are mixed with the second $\beta$ bits of the desired $\beta^2$-bit string pad, etc.

\paragraph{Stage~4.}
Mix all $\beta$ bins (each containing a $\beta$-bit string pads) into a common bin to create $\beta^2$-bit string pads.
The wings ensure that the bit string pads attach in the desired order.
Repeat this process $\frac{x}{\beta^2}$ times, attaching each $\beta^2$-bit string pad onto a growing assembly containing all previous pads, to yield the final $x$-bit string pad.
In total, $\BO{\frac{\log\log \beta}{\log \gamma}}$ stages are used to construct the wings and $\BO{\frac{x}{b^2}}$ stages are used to assemble $\frac{x}{\beta^2}$ unique $\beta^2$-bit string pads.
Thus the total stage complexity is $\BO{\frac{x}{\beta^{2}} + \frac{\log \log \beta}{\log \gamma}} = \BO{\frac{x}{b^{2}} + \frac{\log \log b}{\log t}}$.
\end{proof}

%% file: nxn_square.tex
\section{Assembly of $n\times n$ Squares}\label{sec:nxn}

Efficient assembly of $n \times n$ squares is obtained by combining bit string pads with a technique of Rothemund and Winfree~\cite{RotWin00}.
Their technique utilizes a binary counting mechanism which constructs a length $\Theta(n)$ rectangle with $\Theta(\log{n})$ width.
The mechanism uses $\BO{\log{n}}$ tile types to seed the counter at a certain value, and then $\BO{1}$ tile types attach in a ``zig-zag'' pattern, where ``zigs''  copy the value from the row below and ``zags'' increment the the value by $1$.
Once the binary counter increments to its maximum value (a string of $1$'s), the assembly stops growing.
The length of the resulting rectangle is determined by the seed bit string pad on which binary counting begins.
Two rectangles assembled this way serve as two adjacent sides of the desired square. 
The square can be filled in using these two rectangles as binding locations for generic filler tiles.
We utilize the bit string pad construction of Section~\ref{sec:bit_string_pad} to efficiently assemble the seed for the binary counting mechanism, requiring only an additional $\BO{1}$ tile types and $1$ stage to perform the binary counting and square filling.

\begin{figure}[t]
    \centering
    \includegraphics[width=0.85\textwidth]{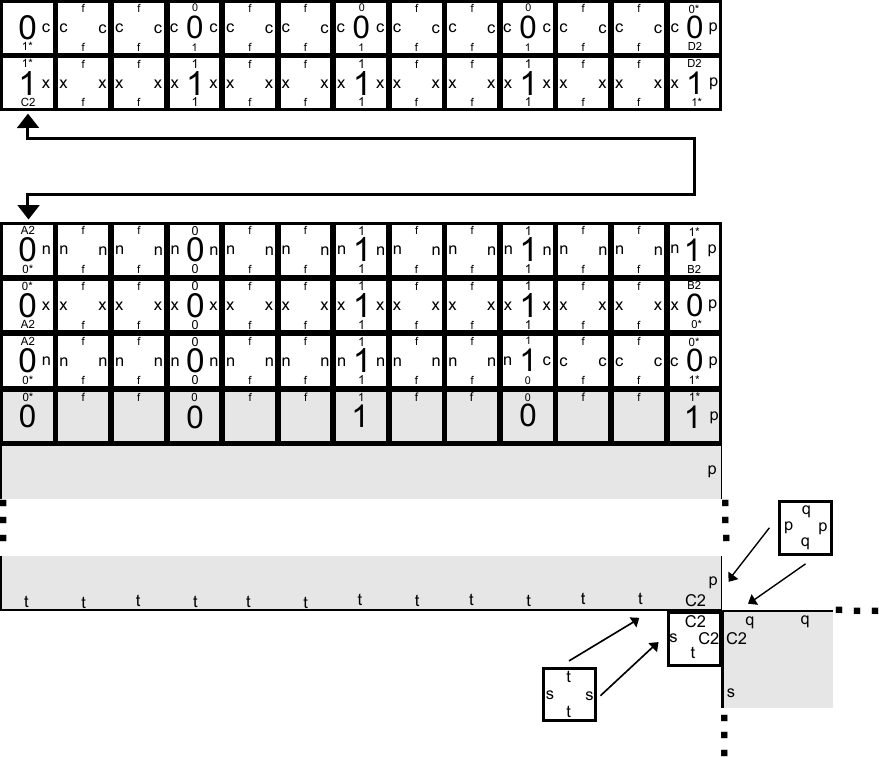}
    \caption{Constructing a counter seed. The bit string pads are shown in gray. Glues with a ``2'' in the string have strength-2, all other glues have strength~1.}
    \label{fig:nxnsquare}
\end{figure}

\begin{theorem} \label{thm:nxn}
There exists a constant $c$ such that for any $b,t,n\in\mathbb{N}$ with $b,t > c$, there exists a $\tau = 2$ staged assembly system with $b$ bins, $t$ tile types, and $\mathcal{O}(\frac{\log n - t\log t - tb}{b^2} + \frac{\log \log b}{\log t})$ stages whose uniquely produced output is an assembly whose shape is an $n\times n$ square.
\end{theorem}

\begin{proof}
Let $c$ be the (constant) number of tile types used to implement the fixed-width ``zig-zag'' binary counting mechanism shown in~\cite{RotWin00}.
Let $t' = t-c$, $b' = b-2$, and $n' = \lceil \log n \rceil$.
Let $m = 2^{n'-1} - (n-22)/2 - n'(2\log b' + 2)$.
Using Lemma~\ref{lem:xbit_string}, construct two $\Theta(\log{b})$-gap $\lceil \log m \rceil$-bit string pads encoding $m$, where each construction each uses $b'$ bins, $t'$ tile types and $\BO{\frac{\lceil \log m \rceil-tb-t\log t}{b^2} + \frac{\log \log b}{\log t}}$ stages.
Figure~\ref{fig:nxnsquare} shows the construction, including modifications to the technique shown in~\cite{RotWin00}.

On both pads, a small modification is made: the glues of the first and last bits are made unique and the first bit's glue strength is set to~2.
This modification is necessary to implement a fixed-width binary counting mechanism as in~\cite{RotWin00} and uses $\BO{1}$ additional tile types.
Also, on the north-facing (east-facing) bit string pad, a unique strength-2 glue ${\rm C2}$ is placed on the south (west) face of the pad's southmost eastern (northmost western) tile.
This special glue is used to combine the two pads with a unique tile type.

Note that the bit string pads assembled in Section~\ref{sec:bit_string_pad} have substantial spacing between the exposed binary glues, but the counter of~\cite{RotWin00} has spacing~0.
This is resolved by adding generic tiles which transfer information horizontally.
These generic tiles use cooperative binding between a south-facing $f$ glue (which matches the glue that spaces the bits on the bit string pad) and west/east glues representing the information to be passed horizontally across spacing of $f$ glues.
The tiles also expose a north-facing $f$ glue to be used when the information needs to be transferred across the spacing in the row above.
Without loss of generality, rotated versions of these tiles are used in the east-growing counter.

The stage complexity of the system is $\BO{\frac{\lceil \log n \rceil-tb-t\log t}{b^2} + \frac{\log \log b}{\log t}}$.
Note that the length of the bit string pads assembled according to Lemma~\ref{lem:xbit_string} is dependent on $b$, the number of bins used to construct the bit string pad.
If $b$ is so large that the spacing between bits causes the width of the bit string pad to exceed $n$ (roughly $\log{b} > n$), we instead directly construct the appropriate bit string pad with spacing~0 using $\BO{\frac{\log \log b}{\log t}}$ stages.
\end{proof}

The following lower bound is derived from Lemma~\ref{lem:info-lowerbound} by observing that for almost all $n \in \mathbb{N}$, $\lfloor \log{n} \rfloor$ bits are needed to represent $n$.

\begin{theorem}
\label{thm:nxnLower}
For any $b,t\in\mathbb{N}$ and almost all $n\in\mathbb{N}$, any staged assembly system which uses at most $b$ bins and $t$ tile types whose uniquely produced output is an assembly whose shape is an $n\times n$ square must use $\Omega(\frac{\log{n} - t\log{t} - tb}{b^2})$ stages.
\end{theorem}

%% file: scaled_shapes.tex
\section{Assembly of Scaled Shapes} \label{sec:scaledShapes}

\begin{figure}[t]
	\centering
        \includegraphics[width=0.43\textwidth]{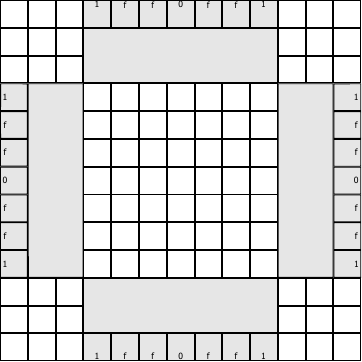}
    \caption{Construction of the modified seed block. Bit string pads are colored in gray.  We concatenate four $K(S)$-bit string pads representing $S$.}
    \label{seedblock}
\end{figure}

Efficient assembly of arbitrary shapes (up to scaling) is achieved by combining bit string pads with the shape-building scheme of Soloveichik and Winfree~\cite{SolWin07}.
Their construction uses two subsets of tile types: a varying set to encode the binary description of the target shape and a fixed set to decode the binary description and build the shape.
We replace the first set with a bit string pad encoding the same information.

\begin{theorem} \label{thm:shapes}
There exists a constant $c$ such that for any $b,t\in\mathbb{N}$ with $b,t > c$ and any shape $S$ of Kolmogorov complexity $K(S)$, there exists a $\tau = 2$ staged assembly system with $b$ bins, $t$ tile types, and $\mathcal{O}(\frac{K(S)-t\log t - tb}{b^2} + \frac{\log \log b}{\log t})$ stages whose uniquely produced output has shape $S$ at some scale factor.
\end{theorem}

\begin{proof}
Observe that the tile set described in~\cite{SolWin07} uniquely constructs the same terminal assembly, namely a \emph{scaled} version of $S$ where each cell is replaced by a square block of cells, when run at temperature~2 in the two-handed mixing model.
It does so via a Kolmogorov-complexity-optimal Turing machine simulation of a machine that computes a spanning tree of the shape given a seed assembly or \emph{seed block} encoding the shape.
The simulation is then run as it ``fills in'' the shape, beginning with the seed block.
Here a similar seed block is constructed and consists of four bit string pads, a square ``core'' and additional filler tiles.

Let $t' = \frac{t-c}{5}$ where $c$ is the (constant) number of tile types required by~\cite{SolWin07} to carry out the simulation of a (fixed) universal Turing machine.
Let $b' = \frac{b-1}{5}$.

Use the method of Lemma~\ref{lem:xbit_string} to construct the modified seed block by assembling four different $K(S)$-bit string pads representing a program that outputs $S$, each using $b'$ bins, $t'$ tile types and $\BO{\frac{K(S) - t'\log t' - t'b'}{b'^2} + \frac{\log \log b'}{\log t'}}$ stages.
These four pads (each with dimensions $(2K(S) \log{K(S)}+2) \times \BO{1}$) are attached to the four sides of a $\left ( 2K(S)\log K(S)+2 \right ) \times \left ( 2K(S)\log K(S)+2 \right )$ square constructed as in Theorem~\ref{thm:nxn} using $t'$ tile and $b'$ bins in $\BO{\frac{\log(2K(S)\log K(S)+2) - t'\log t' - t'b'}{b'^2} + \frac{\log \log b'}{\log t'}}$ stages.
An abstract figure of the completed seed block can be seen in Figure~\ref{seedblock}.
The Turing simulation occurs in one stage by mixing the four concatenated bit string pads into one bin that contains the fixed set of tile types simulating a Turing machine described in~\cite{SolWin07}.
The bit string pads contain spacing between the exposed binary glues, while the simulation tile types of~\cite{SolWin07} expect adjacent glues.
This is resolved by modifying the Turing-machine-simulation tile set to include generic tiles for transferring information across spacing, similar to the tiles of the same purpose discussed in the proof of Theorem~\ref{thm:nxn}. We need at most $1$ such tile for each tile in the (constant-sized) Turing-machine-simulation tile set, for a constant increase in tile complexity.
The stage complexity is $4\times\mathcal{O}(\frac{K(S) - t'\log t' - t'b'}{{b'}^2} + \frac{\log \log b'}{\log t'}) + \mathcal{O}(\frac{\log( 2K(S)\log K(S)+2 ) - t'\log t' - t'b'}{b'^2} + \frac{\log \log b'}{\log t'}) = \mathcal{O}(\frac{K(S) - t\log t - tb}{b^2} + \frac{\log \log b}{\log t})$.
\end{proof}

The following theorem follows from the information-theoretic bound of Lemma~\ref{lem:info-lowerbound}.

\begin{theorem}
\label{thm:shapesLower}
For any $b,t\in\mathbb{N}$ and shape $S$ with Kolmogorov complexity $K(S)$, any staged assembly system which uses at most $b$ bins and $t$ tile types whose uniquely produced output has shape $S$ must use $\Omega(\frac{K(S) - t\log t - tb}{b^2})$ stages.
\end{theorem}

%% file: flexible_glues.tex
\section{Flexible Glues}
\label{sec:flexible_glues}

Here, we consider reducing stage complexity using \emph{flexible glues} that allow non-equal glue pairs to have positive strength.
These can be used to encode $\Theta(t^2)$ bits rather than $\Theta(t\log{t})$ bits in $t$ tile types, reducing both the upper and lower stage complexity bounds.
For the upper bound, the additional bits are encoded by  modifying Step~2 of the bit string pad construction of Section~\ref{sec:bit_string_pad} to use a modified decompression bad, similar to the technique introduced in~\cite{AGKS05g}.

\begin{definition}[Flexible decompression pad]
A width-$k$ length-$d^2$ flexible decompression pad is a $k\times d^2$ rectangular assembly with $d^2$ north glue types from the set $\{{g_{\rm start}}, g_0, g_1, \dots, g_d\}$ exposed on the north face of the rectangle.
The westmost glue is $g_{\rm start}$, the following $r-1$ glues have type $g_0$, followed by $d$ glues of type $g_1$, $d$ glues of type $g_2$, and so on.
The exposed south, east, and west tile edges have glues $g_S$, $g_E$, and $g_W$, respectively.
\end{definition}

\begin{figure}[ht]
    \centering
    \includegraphics[width=1.0\textwidth]{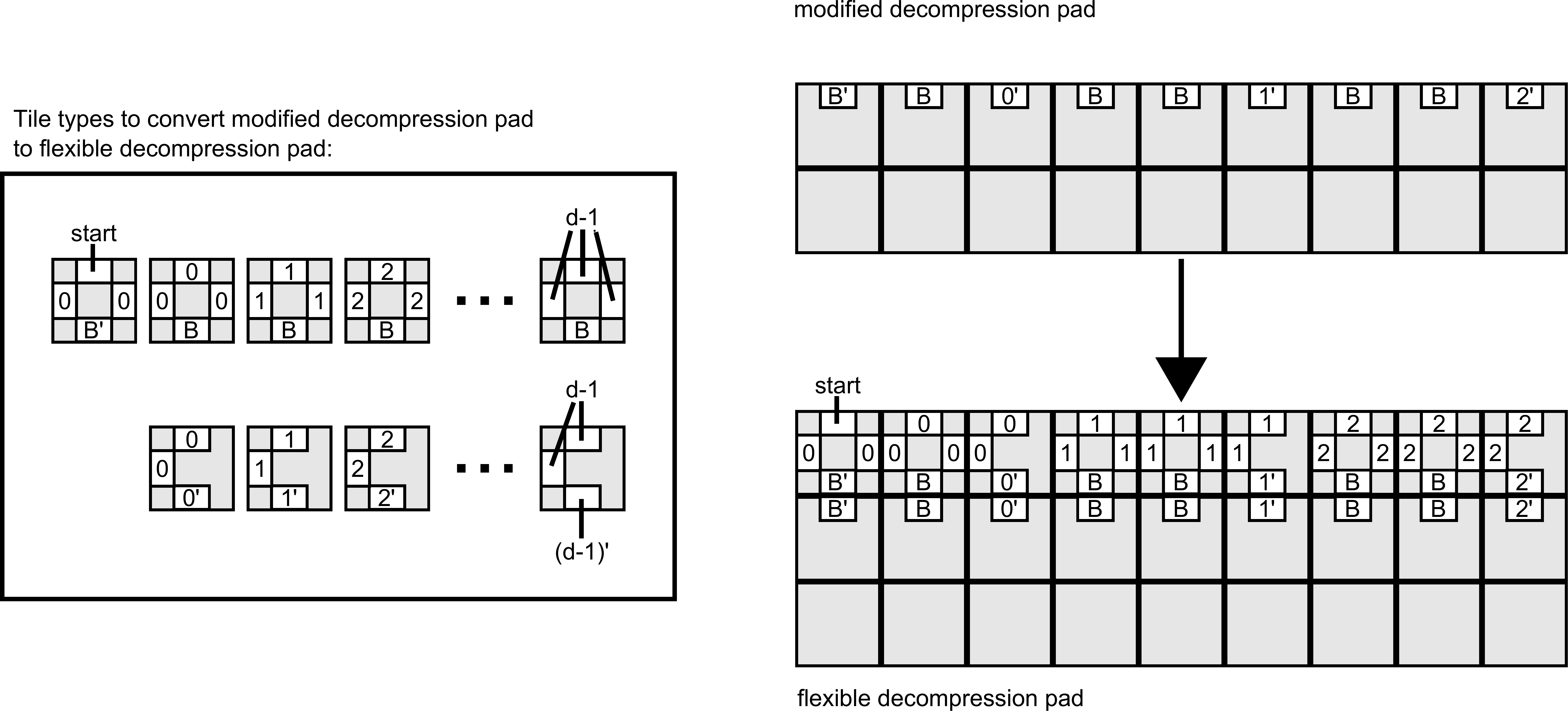}
    \caption{The templates to convert a modified decompression pad to a flexible decompression pad using $2d+1$ tile types, where integer $d \ge 3$, on the left. Using these additional tile types, a modified decompression pad is converted into a flexible decompression pad. A \emph{modified decompression pad} has a westmost northmost glue of $g_{B'}$ and every non-$g_B$ glue on the north surface is a special \emph{prime} version distinct from other similar glue types. On the left, a width-2 modified decompression pad representing the string $012$ in base-8 is converted to a width-3 length-9 flexible decompression pad.}
    \label{fig:fx_conversion}
\end{figure}

In order to build the flexible decompression pad, a modified decompression pad representing a number $C = c_0 c_1 \dots c_{d-1}$ in base $2^d$ is needed.

\begin{lemma}\label{lem:fx_decompression_pad}
Given an integer $d \ge 3$, there exists a $\tau = 2$ staged assembly system with 1 bin, $8d-1$ tile types, and 1 stage that assembles a width-3 length-$d^2$ flexible decompression pad.
\end{lemma}

\begin{proof}
Start with the construction of Lemma~\ref{lem:decompression_pad} that yields a a width-2 length-$d^2$ decompression pad encoding $C$.
Modify the tile types of this construction such that the westmost northern glue is $g_{B'}$ and every non-$g_B$ glue on the north surface is a special \emph{prime} version, to differentiate between other similar glue types.
Then add $2d+1$ tiles that modify the north surface decompression pad to yield width-3 flexible decompression pad, as seen in Figure~\ref{fig:fx_conversion}.
This step requires $5d + \log{2^d}-2 $ tile types to build a modified decompression pad and $2d+1$ tiles to convert this modified decompression pad that into a flexible decompression pad.
Thus $2d + 1 + 5d + \log{2^d}-2 \leq 8d-1$ tile types are used in total.
\end{proof}

\begin{figure}[ht]
    \centering
    \includegraphics[width=0.9\textwidth]{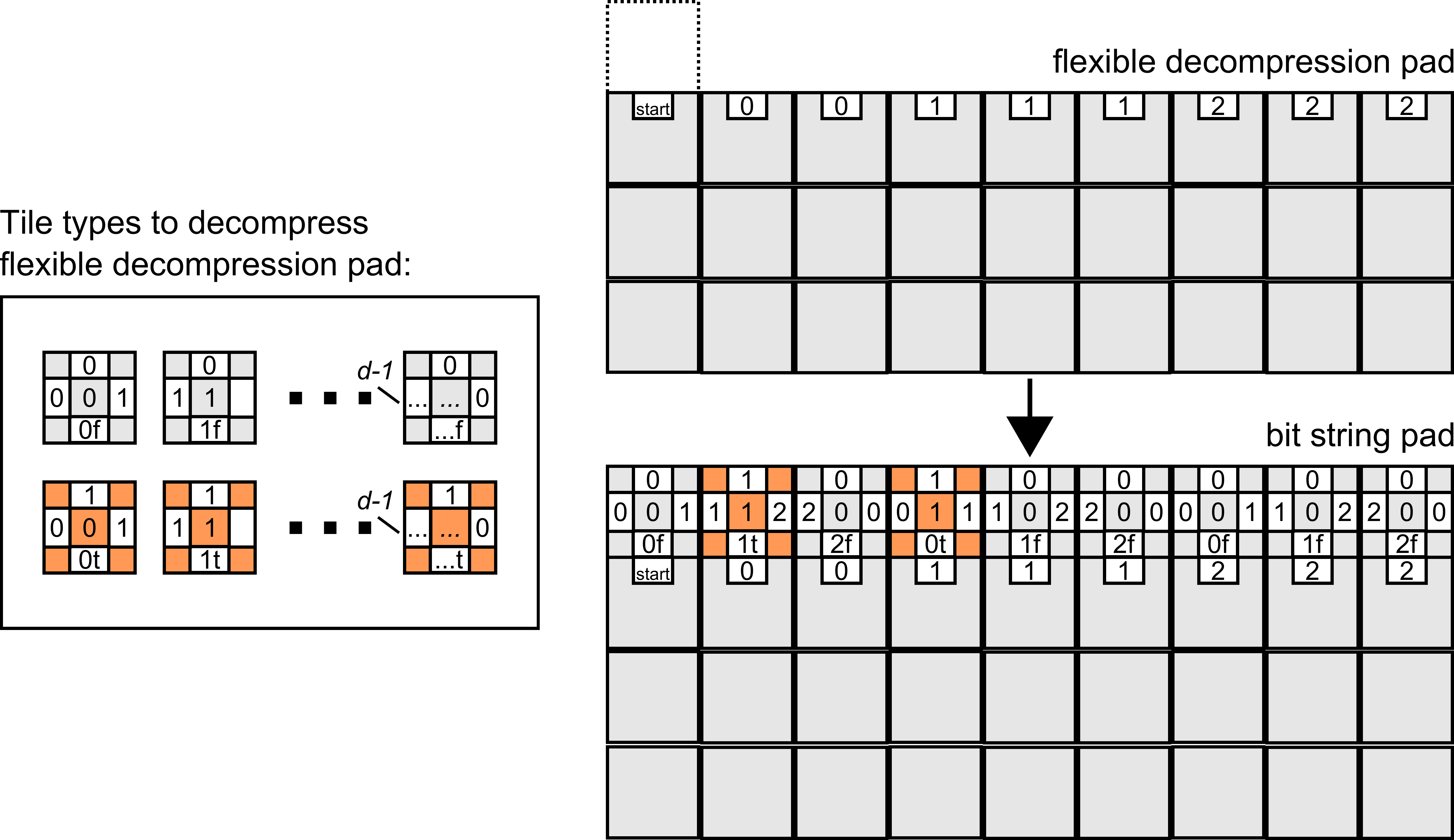}
    \caption{On the left, the templates for the decompression tiles needed to decompress a flexible decompression pad for any given $d \ge 3$. In the top right, an example of a length-9 flexible decompression pad. In the bottom right, the decompression tiles interact with the flexible decompression pad and glue function to assemble a bit string pad from a flexible decompression pad, representing the bit string $010100000$. The flexible glues form a bond of strength~2 between the glue pair ($g_{start}$, $g_{0f}$), strength~1 between glues pairs ($g_{0}$, $g_{0}$), ($g_{1}$, $g_{1}$), ($g_{2}$, $g_{2}$), ($g_{0}$, $g_{1t}$), ($g_{0}$, $g_{2f}$), ($g_{1}$, $g_{0t}$), ($g_{1}$, $g_{1f}$), ($g_{1}$, $g_{2f}$), ($g_{2}$, $g_{0f}$), ($g_{2}$, $g_{1f}$), and ($g_{2}$, $g_{2f}$), and strength~0 between all other glue pairs.}
    \label{fig:fx_decompression_example}
\end{figure}

\begin{lemma}
\label{lem:step_2_1_fx}
Given integers $d\ge 3$ and any length $d^2$ bit string $R$, there exists a 1 stage, 1 bin, $\tau = 2$ staged assembly system with flexible glues that assembles a width-4 gap-0 $d^2$-bit string pad representing $R$, using at most $10d-1$ tile types.
\end{lemma}

\begin{proof}
Consider a width-3 length $d^2$ flexible decompression pad.
The idea is to use $2d$ tile types and flexible glues to build a width-4 gap-0 $d^2$-bit string pad from the flexible decompression pad (see Figure~\ref{fig:fx_decompression_example}).
Consider a sequence of $d$ bitstrings $D = D_0, D_1, \dots, D_{d-1}$ with each $D_i = s_0 s_1 s_2 \dots s_{d-1}$ such that the in-order concatenation of all bitstrings in $D$ equals $R$.
Let $D_{i,j}$ denote the $j$th bit of the $i$th bit string of $D$.

The goal is to construct a glue function such that it specifies the tiles that can attach to the top of the flexible decompression pad to be the concatenation of the bitstrings in $D$.
Note that the tiles that have a ``0'' or ``1'' label as those with glues that end in ``f'' or ``t'', respectively.
Let $str(g, g')$ denote the strength between any two glues $g$ and $g'$.
Set the tile that attaches to the $g_{\rm start}$ glue to be one that exposes $g_0$ or a $g_1$ by setting $str(g_{\rm start},g_{0f})=2$ or $str(g_{\rm start},g_{0t})=2$, respectively.
For all $D_{i,j}$, we set $str(g_i,g_{jf})=0$ if and only if $D_{i,j}=0$ and $str(g_i,g_{jt})=1$, otherwise.
In addition, we set $str(g_0,g_0)=1$, $str(g_1,g_1)=1$, and so on, up to $str(g_d, g_d)=1$.

With this, we build a width-4 gap-0 $d^2$-bit string pad from the flexible decompression pad.
An example of this can be seen in Figure~\ref{fig:fx_decompression_example}.
Also, $8d-1$ tile types are used to build a width-3, length $d^2$ flexible decompression pad.
An additional $2d$ tile types are needed to decompress, using flexible glues, into a width-4 $d^2$-bit string pad.
So the total number of tile types used is $10d-1$.
\end{proof}

\begin{lemma}
\label{lem:step_2_2_fx}
There exists a constant $c$ such that for any $t\in\mathbb{N}$ with $t > c$ there exists an $x = \Theta(t^2)$ such that for any bit string $S$ of length $x$, there exists a  $\tau = 2$ staged assembly system with flexible glues with $1$ bin, $t$ tile types and $1$ stage whose uniquely produced output is a width-$4$ gap-$0$ $x$-bit string pad representing $S$.
\end{lemma}

\begin{proof}
Given $t$ tile types, consider how many bits can be produced using Lemma~\ref{lem:step_2_1_fx}.
Let $d=\lfloor \frac{t+1}{10} \rfloor$.
Invoke Lemma~\ref{lem:step_2_1_fx} to build a width-4 $d^2$-bit string pad with flexible glues using $10d-1$ tiles.
The number of bits produced is $y = d^2 = (\lfloor \frac{t+1}{10} \rfloor)^2 = \Theta(t^2)$.
Then by Lemma~\ref{lem:step_2_1_fx}, any width-4 $(\lfloor \frac{t+1}{10} \rfloor)^2$-bit string pad can be built in the flexible glue model using at most $t$~tiles, 1~stage, and 1~bin.
The smallest choice of $d$ requires $d=\lfloor \frac{t+1}{10} \rfloor \ge 3$, implying $t \geq 29$.
For all cases where $t \geq c$ we have a constant, $c=29$, where this lemma holds true.
\end{proof}

The improvements to Lemmas~\ref{lem:step_2_1_fx} and~\ref{lem:step_2_2_fx} allow for a larger bit string pad to be built in Step~2 when compared to standard glues, reducing stage complexity to $\mathcal{O}(\frac{\log \log b}{\log t} + \frac{x-tb-t^2}{b^2})$:

\begin{lemma}
\label{lem:xbit_stringFlex}
There exists a constant $c$ such that for any $b,t \in \mathbb{N}$ with $b,t > c$ and a bit string $S$ of length $x$, there is a $\tau = 2$ staged assembly system with flexible glues with $b$ bins, $t$ tile types, and $\mathcal{O}(\frac{x-tb-t^2}{b^2} + \frac{\log \log b}{\log t})$ stages whose uniquely produced output is a width-$9$ gap-$\Theta(\log b)$ $x$-bit string pad representing $S$.
\end{lemma}

Nearly tight upper and lower bounds for square and general shape construction in the flexible glue model are obtained by replacing the bit string construction of Lemma~\ref{lem:xbit_string} with Lemma~\ref{lem:xbit_stringFlex}, and applying the flexible glue lower bound of Lemma~\ref{lem:info-lowerbound}:

\begin{theorem} \label{thm:nxnFlex}
There exists a constant $c$ such that for any $b,t,n\in\mathbb{N}$ such that $b,t > c$, there exists a $\tau = 2$ staged assembly system with flexible glues with $b$ bins, $t$ tile types and $\mathcal{O}(\frac{\log n - t^2 - tb}{b^2} + \frac{\log \log b}{\log t})$ stages whose uniquely produced output is an $n\times n$ square.
\end{theorem}

\begin{theorem}
\label{thm:nxnLowerFlex}
For any $b,t\in\mathbb{N}$ and almost all $n\in\mathbb{N}$, any staged assembly system with flexible glues which uses at most $b$ bins and $t$ tile types whose uniquely produced output is an $n\times n$ square must use $\Omega(\frac{\log{n} - t^2 - tb}{b^2})$ stages.
\end{theorem}

\begin{theorem} \label{thm:shapesFlex}
There exists a constant $c$ such that for any shape $S$ and $b,t\in\mathbb{N}$ with $b,t > c$, there exists a $\tau = 2$ staged assembly system with flexible glues with $b$ bins, $t$ tile types, and $\mathcal{O}(\frac{K (S)-t^2 - tb}{b^2} + \frac{\log \log b}{\log t})$ stages whose uniquely produced output is $S$ at some scale factor.
\end{theorem}

\begin{theorem}
\label{thm:shapesLowerFlex}
For any $b,t\in\mathbb{N}$ and shape $S$ with Kolmogorov complexity $K(S)$, any staged assembly system with flexible glues which uses at most $b$ bins and $t$ tile types whose uniquely produced output is $S$ must use $\Omega(\frac{K(S) - t^2 - tb}{b^2})$ stages.
\end{theorem}

%% file: conclusion.tex
\section{Conclusion}\label{sec:conclusion}

In this work, we achieve nearly optimal staged assembly of two classic benchmark shape classes.
These constructions generalize the known upper bounds of~\cite{AdChGoHu01,AGKS05g,DDFIRSS07,SolWin07} to all combinations of tile type and bin complexity, as well as to the flexible glue model.

The obvious technical question that remains is whether the additive $\BO{\frac{\log{\log{b}}}{\log{t}}}$ gap between the upper and lower bounds can be removed.
This gap in our constructions is induced by the wings subconstruction of Section~\ref{sec:wings}, and seems difficult to eliminate, as the wings serve as our generic solution to assembly labeling and coordinated attachment.
As such, the wings subconstruction might be useful for improving the efficiency of staged assembly for other shape classes.

Finally, the constant width of the bit string pads we construct was not exploited here, but may be useful for assembling shapes with geometric bottlenecks, e.g., thin rectangles.